\documentclass{article}
\usepackage[margin=0.5in]{geometry}
\usepackage[utf8]{inputenc}
\usepackage{amsfonts,graphicx,amssymb}
\usepackage{amsmath,amsthm}
\usepackage{graphicx}
\usepackage{epstopdf}
\newcommand{\e}[1]{\begin{equation}#1\end{equation}}
\newcommand{\enn}[1]{\begin{equation*}#1\end{equation*}}
\newcommand{\tab}[2]{\begin{array}{#1}#2\end{array}}
\newcommand{\syst}[2]{\e{\left\{\tab{#1}{#2} \right. }}
\newcommand{\sysnn}[2]{\enn{\left\{\tab{#1}{#2} \right. }}
\newcommand{\Ok}{{\Omega_k}}

\newcommand{\Oj}{{\Omega_j}}
\newcommand{\bOk}{{\partial \Omega_k}}
\newcommand{\skj}{{\Sigma_{kj}}}
\newcommand{\sjk}{{\Sigma_{jk}}}
\newcommand{\Gk}{{\Gamma_k}}

\newcommand{\Vpsi}{\frac{\nabla \psi}\epsilon}
\newcommand{\Vphi}{\frac{\nabla \phi}\epsilon}
\newcommand{\Vthe}{\frac{\nabla \theta}\epsilon}

\newcommand{\dpsi}{\partial_\psi}
\newcommand{\dphi}{\partial_\phi}
\newcommand{\dthe}{\partial_\theta}
\newcommand{\eun}{\mathbf e_1}
\newcommand{\ede}{\mathbf e_2}
\newcommand{\etr}{\mathbf e_3}

\newcommand{\unNh}{[\![ 1, N_h ]\!]}

\newcommand{\dx}{\partial_x}
\newcommand{\dy}{\partial_y}
\newcommand{\lu}[2]{\lambda_{#1,#2}}
\newcommand{\luoz}{\lambda_{1,0}}
\newcommand{\luzo}{\lambda_{0,1}}
\newcommand{\luij}{\lambda_{i,j}}

\newtheorem{hp}{Hypothesis}[section]
\newtheorem{definition}{Definition}[section]
\newtheorem{lemma}{Lemma}[section]
\newtheorem{theorem}{Theorem}[section]

\begin{document}




\title{Well-posedness and generalized plane waves simulations of a 2D mode conversion model}


\author{Lise-Marie Imbert-G\'erard\footnote{imbertgerard@cims.nyu.edu, Courant Institute of Mathematical Sciences, New York University, 251 Mercer street,
New York, NY 10012}
}
\maketitle

\begin{abstract}
 Certain types of electro-magnetic waves propagating in a plasma can undergo a mode conversion process. In magnetic confinement fusion, this phenomenon is very useful to heat the plasma, since it permits to transfer the heat at or near the plasma center. This work focuses on a mathematical model of wave propagation around the mode conversion region, from both theoretical and numerical points of view. It aims at developing, for a well-posed equation, specific basis functions to study a wave mode conversion process. These basis functions, called generalized plane waves, are intrinsically based on variable coefficients. As such, they are particularly adapted to the mode conversion problem. The design of generalized plane waves for the proposed model is described in detail. Their implementation within a discontinuous Galerkin method  then provides numerical simulations of the process. These first 2D simulations for this model agree with qualitative aspects studied in previous works.
\end{abstract}





\section{Introduction}
Mode conversion corresponds to a transfer of energy between different types of propagating waves. It is an important problem in magnetic confinement fusion particularly for plasma heating or current drive. Indeed, some waves used for these applications cannot propagate  directly from the launching region at the wall toward the point of the plasma where they would be useful. But some other forms of wave can be sent from the wall toward the plasma, to be converted into the desired wave at the mode conversion region, and so penetrate until the heating or control point. See   \cite{Kop} for a first study of mode conversion equations.
Even though experimental models \cite{Bonoli,Laq1,Laq2,Laq3} and simple one dimensional models \cite{Bonoli2,Volpe} have been studied, the two dimensional mathematical model is still not well understood. In  \cite{HaroldConv} a two dimensional model is derived by means of an asymptotic expansion, but  the resulting equation is not standard in the literature ; propagating solutions are then constructed thanks to an integral representation.

Mode conversion corresponds to a propagating wave transmitted from one propagative zone to another one, even though the two zones only touch along a curve. The mode conversion region is defined as the vicinity of this curve.
The two dimensional model studied in the present work comes from the cold plasma model for wave propagation in a plasma confined by a magnetic field, and reads
\begin{equation}\label{eq:2ndOFnt}
\left( \dx^2 +(d+\overline d)\dx\dy +|d|^2\dy^2\right) F + (d-\overline d)x\partial_y F-\left(1+\frac{1}{\mu}+x(x+y)\right)F = 0,
\end{equation} 
where $F$ is the scalar unknown and $d$ and $\mu$ are complex parameters linked to the confining magnetic field and the electron density.
 Most of the derivation will follow the steps of   \cite{HaroldConv}. However a different exposition is proposed here, highlighting the different steps of the reasoning and insisting on a crucial change of unknowns.  Moreover, Equation \eqref{eq:2ndOFnt} does not appear in   \cite{HaroldConv}, where an equivalent equation is given for a different unknown. Compared to the latter, Equation \eqref{eq:2ndOFnt} has a particular structure, more convenient to prove the well-posedness. This work presents the first well-posedness result for a mode conversion equation obtained by expansion in the mode conversion region.
 In this elliptic second order linear 2D partial differential equation,  the physical properties of the domain appear in the zeroth order term, as the sign of $f(x,y) = x(x+y)$. The domain is propagative on $\{ (x,y) / f(x,y)<0\}$ and absorbing on $\{ (x,y) / f(x,y)>0\}$, see Figure \ref{fig:x(x+y)i} and the mode conversion region is reduced to the vicinity of the origin. It is clear that mode conversion is strongly linked to variable coefficients.  
\begin{figure}
\begin{center}
\includegraphics[height=6cm]{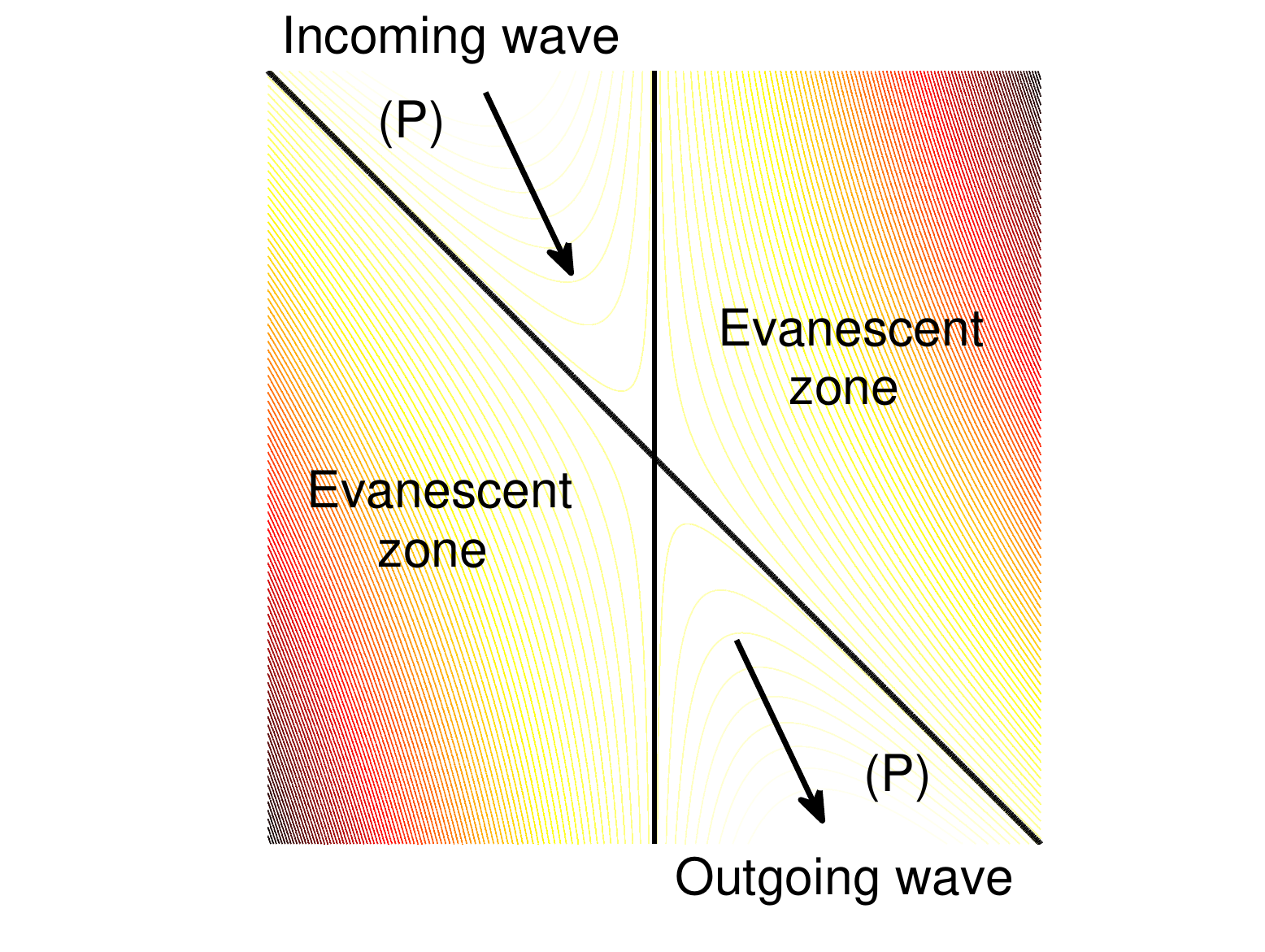}
\end{center}
\caption{Level curves of the function $f(x,y) = x(x+y)$, indicating the propagative (P) and evanescent zones as well as the two cut-offs lines, $x=0$ and $x+y=0$, limiting the different zones. They cross at the mode conversion point $(0,0)$.}
\label{fig:x(x+y)i}
\end{figure} 

This work aims at developing Generalized Plane Waves (GPWs) to study a wave mode conversion process. 
GPWs have been developed following the idea that plane waves are relevant basis functions to solve wave problems numerically, since the oscillatory behavior of the problem is embedded in their definition. In the same way, GPWs encode information about the problem to be solved, but they are specifically adapted to problems with variable coefficients. They were introduced in   \cite{LMIGBD}, and their interpolation properties for Helmholtz equation with a variable coefficient were presented in   \cite{LMIG}. In this work, for the first time, GPWs are designed for a mode conversion equation, namely \eqref{eq:2ndOFnt}, and are implemented in a discontinuous Galerkin method for numerical simulation of the mode conversion process in this model.

First numerical evidences of a mode conversion process are displayed for this 2D full-wave expansion model. A typical test case is proposed together with a way to estimate the transmission coefficient across the mode conversion region. The influence of different parameters is illustrated, mainly the influence of the incident angle already studied in \cite{Kop}, in a 1D model in   \cite{Volpe} and in the 2D propagating solutions in   \cite{HaroldConv}.

Section \ref{sec:mod} presents the derivation of the second order equation \eqref{eq:2ndOFnt}, describing the mode conversion process, while the existence and uniqueness of a weak solution to this equation are proved in Section \ref{sec:TS}. Even though Section \ref{sec:mod} is self contained and does not require any prior knowledge on waves in plasma, it is independent of the rest of the article. Therefore it can be skipped by a reader willing to start from Equation \eqref{eq:2ndOFnt}. This work then focuses on numerical aspects. Section \ref{apps} first describes a discontinuous Galerkin method for numerical simulation. Then GPWs adapted to the mode conversion equation are carefully designed, distinguishing between the individual construction of an approximated solution to the equation and the global features of a set of independent GPWs. Numerical results are finally displayed in Section \ref{sec:NS}, showing evidences of mode conversion and highlighting the influence of different parameters on the process.

\section{A wave propagation model in the mode conversion region}
\label{sec:mod}
This section presents the formal process leading to the equation studied in the rest of this article. 
It is mainly based on   \cite{HaroldConv}, since Equation \eqref{eq:2ndOF} is Equation (59) from that reference. The goal of this presentation is to shed a different light on the derivation of the equation.
For the sake of completeness, preliminary material related to plasma physics and the wave propagation model is presented first.
The following subsection then focuses on Maxwell's equations in the mode conversion region:
since the dielectric tensor is defined in a simpler way in a specific orthogonal coordinate system, the idea is to obtain a reduced system with fewer components of the electric and magnetic fields in those coordinates. Such a simplification follows naturally from the techniques of geometrical optics, see  \cite{Herman}. It is based on physically relevant hypotheses concerning the relative order of the studied quantities, and subsequent expansions valid in the mode conversion region. Consider $\epsilon$, the inverse of the classical geometrical optics expansion parameter.
 If $\omega$ is the wave frequency, $c$ the speed of light in vacuum and $L$ the characteristic equilibrium wavelength, $\epsilon$ is defined as $\epsilon = c/(\omega L)$.
  The model will result from a formal expansion with respect to $\epsilon$. The scaling assumptions follow the perpendicular stratified case, as studied in   \cite{WB}.

\subsection{Preliminaries}\label{ssec:Prel}

The toroidal geometry can be described  by the classical axisymmetric coordinates $(r,\phi,z)\in \mathbb R^+\times[0,2\pi)\times \mathbb R$, see Figure \ref{fig:Coord}, 
and the corresponding orthonormal right handed basis $(\mathbf e_r,\mathbf e_\phi, \mathbf e_z)$. The poloidal plane is defined as a half plane given by a constant $\phi$, so that $(\mathbf e_r,\mathbf e_z)$ is an orthonormal basis of the poloidal plane while $\nabla \phi$ is orthogonal to the poloidal plane, and $\mathbf e_\phi = r\nabla \phi$ defines the toroidal direction. However, since the wave propagation phenomena are driven by the confining magnetic field, a toroidal coordinate system adapted to the magnetic field will be useful to describe the wave propagation model.
\begin{figure}
\begin{center}
\includegraphics[clip,trim=7cm 3cm 6cm 3cm,width=.47\textwidth]{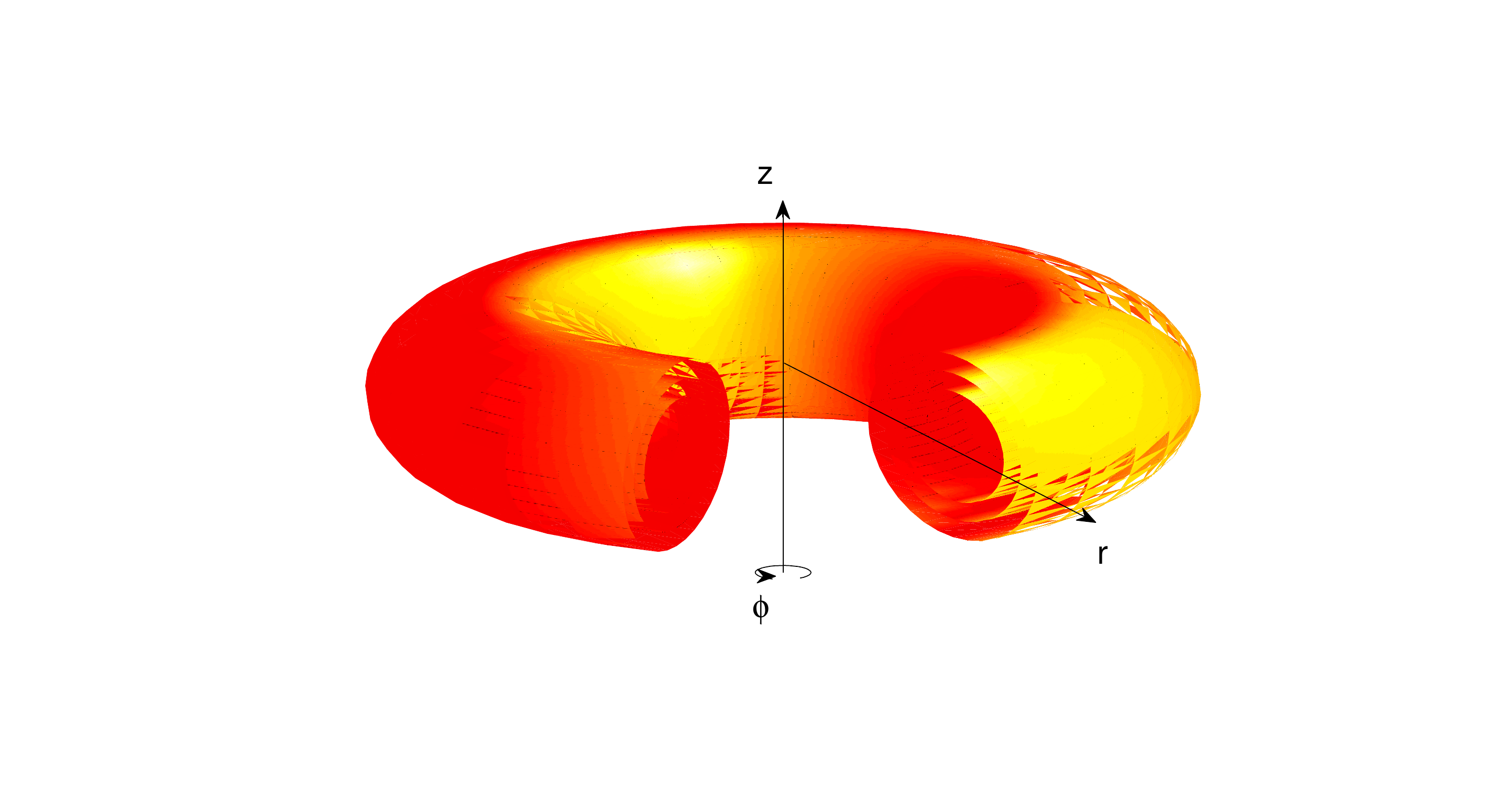}
\includegraphics[clip,trim=3cm .5cm 2.5cm .5cm,width=.47\textwidth]{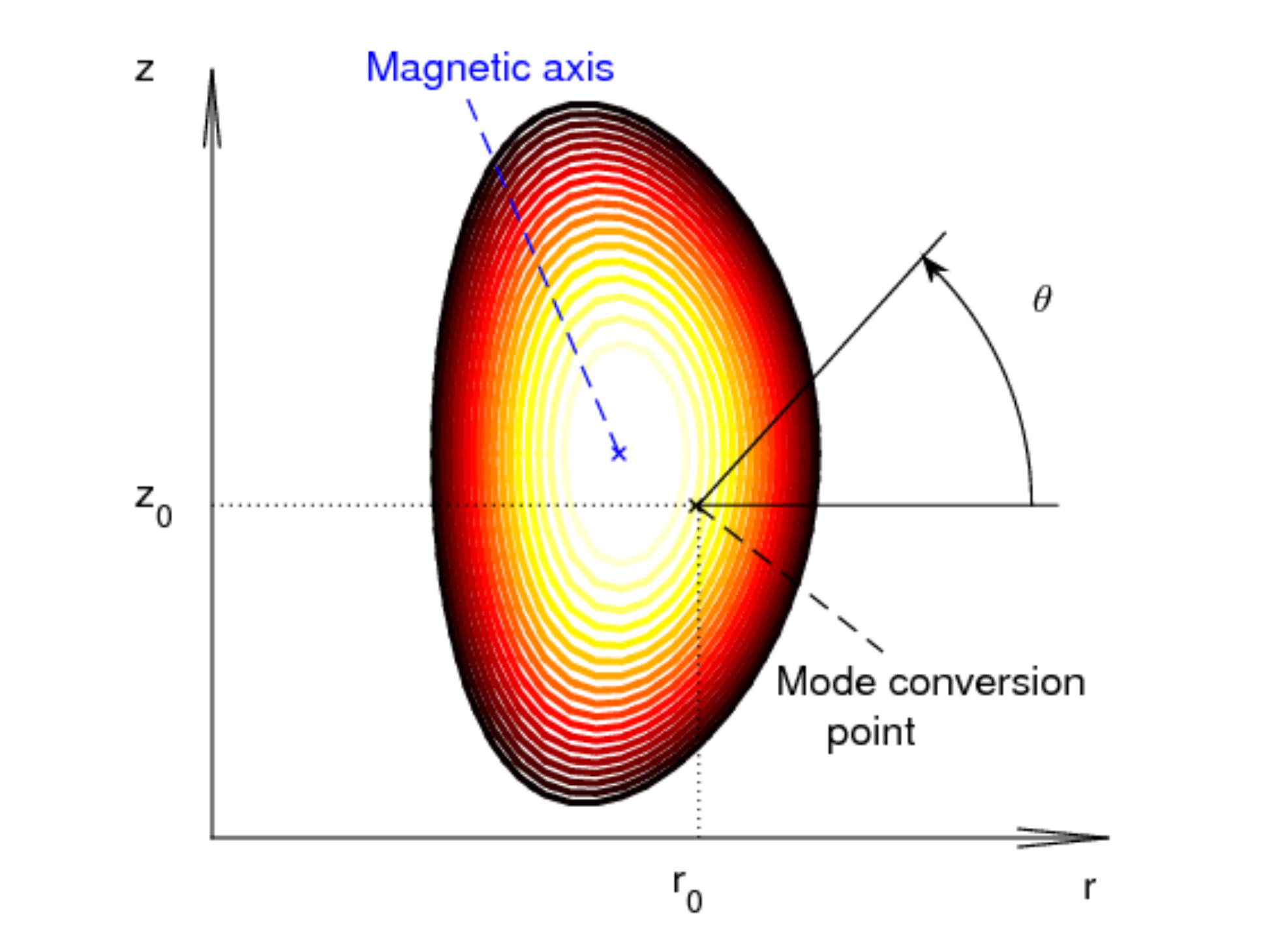}
\end{center}
\caption{Left: Axisymmetric coordinates and toroidal geometry. Right: Poloidal angle $\theta$ and level curves of the flux label $\psi$ in a poloidal half plane. In the poloidal plane, the magnetic axis is the point enclosed by all the flux surfaces. The mode conversion point is distinct from the magnetic axis.}
\label{fig:Coord}
\end{figure}  
We assume that the magnetic field is known from an independent solution of MHD equilibrium. The magnetic field lines display a helical shape, winding around the interior of the torus and that the flux surfaces are closed and nested. In such a case, flux coordinates form a set of coordinates adapted to the shape of the flux surfaces of the confining magnetic field. They consist of a radius-like coordinate, the flux label $\psi\in[\psi_{\min},\psi_{\max}]$, and two angle-like coordinates, the toroidal angle $\phi\in[0,2\pi)$ and poloidal angle $\theta\in[0,2\pi)$. The bounds on $\psi$, $\psi_{\min}$ and $\psi_{\max}$, respectively correspond to the magnetic axis and the outermost closed flux surface. The toroidal angle $\phi$ is the same as the axisymmetric coordinates angle, and the coordinates in the poloidal plane are described in Figure \ref{fig:Coord}. 
The magnetic flux $\psi$ labeling a curve in the poloidal plane measures the flux of magnetic field across the surface enclosed by the curve. It is increasing from the magnetic axis to the boundary of the plasma, and is such that $\nabla\psi$ is orthogonal to the confining magnetic field $\mathbf b$.
 The poloidal angle $\theta$ is such that  $ (\nabla \psi,\nabla \theta)$ is a non-orthogonal basis of the poloidal plane, with the same orientation as $(\mathbf e_r,\mathbf e_z)$. The resulting covariant basis $(\nabla \psi,\nabla  \theta,\nabla  \phi)$ is a left handed non-orthogonal non-normalized basis, associated with the flux coordinates $(\psi,\theta,\phi)$.

The mode conversion region, as the intersection of two cut-off surfaces, has been introduced as a curve. It intersects each poloidal plane at a single point, the point $(r,z) = (r_0,z_0)$ on Figure \ref{fig:Coord}, 
 which stands off the magnetic axis. The mode conversion point is the origin of the flux coordinate system in the poloidal plane,that is to say $( \psi(r_0,z_0),\theta(r_0,z_0))=(0,0)$. Since the origin is not the magnetic axis, at the origin $( \psi,\theta)=(0,0)$ there is no problem to define $\nabla \psi$ while $\nabla \theta$ is not defined. The model will be derived based on a formal expansion with respect to the small parameter $\epsilon$, the inverse of the classical geometrical optics expansion parameter. Following    \cite{HaroldConv}, the mode conversion region is characterized by variations of the poloidal coordinates $(\psi,\theta)$ of order $\sqrt \epsilon$, while the variations of the toroidal angle $\phi$ are of order $1/\epsilon$ as well as $\epsilon r = O(1)$. In this regime the poloidal coordinates can be written as $(\psi,\theta) = \big( \psi(\epsilon r,\epsilon z),\theta(\epsilon r,\epsilon z) \big)$ where $\psi$ and $\theta$ have derivatives of order $O(1)$, so that the derivatives with respect to the flux coordinates will be scaled as $\epsilon(\dpsi,\dthe,\dphi)$, and likewise the gradient vectors will be scaled as $(\Vpsi,\Vthe,\Vphi)$.


The curl operator is crucial in electromagnetics. In order to express it in the flux coordinate system in a concise way, consider the scaled contravariant basis associated with the $(\psi,\theta,\phi)$ coordinates, defined by
\enn{\left( \Vthe\times\Vphi,\Vphi\times\Vpsi,\Vpsi\times\Vthe \right).}
The confining magnetic field $\mathbf b$ has a poloidal $\mathbf b_p$ and a toroidal $\mathbf b_t$ components, such that $\mathbf b = \mathbf b_p+\mathbf b_t$. They are defined by
\enn{\mathbf b = \Vpsi\times\Vphi - Q \Vpsi\times\Vthe , \text{ where } \mathbf b_p = \Vpsi\times\Vphi \text{, } \mathbf b_t = - Q \Vpsi\times\Vthe,}
and the safety factor only depends on the magnetic flux, that is to say $Q=Q(\psi)$\footnote{The safety factor measures the winding of the magnetic field lines around the torus. It is proportional to the ratio between the toroidal and poloidal fields $\frac{|\mathbf b_t|}{|\mathbf b_p|}$. The word safety refers to the resulting stability of a configuration, since a high $Q$ tends to improve the stability and therefore the safety.}.
For clarity $b$, $b_t$ and $b_p$ will respectively denote $|\mathbf b|$, $|\mathbf b_t|$ and $|\mathbf b_p|$. 
The Jacobian of the covariant basis thus reads $\Vpsi\times\Vthe\cdot\Vphi=-b_t/(\epsilon r Q) $. 
The contravariant basis also provides an exact expression of the $curl$ operator:
\enn{
\tab{rl}{
\nabla \times\mathbf V =&
\displaystyle\epsilon\left( 
              \Vthe\times\Vphi  (\dthe V_\phi - \dphi V_\theta) 
              + \Vphi\times\Vpsi  (\dphi V_\psi - \dpsi V_\phi) \right.
\\&\displaystyle \phantom{\epsilon \bigg( 00}\left. 
              + \Vpsi\times\Vthe  (\dpsi V_\theta - \dthe V_\psi)
                                         \right).}}
Another basis will be used to give a simple expression of the dielectric tensor. This other basis is right handed and orthonormal, and is defined by a first vector aligned with $\nabla \psi$ while the third vector is aligned with the magnetic field $\mathbf b$:
\enn{\eun = \Vpsi \left|\Vpsi \right|^{-1},\quad \ede = \mathbf b \times \eun /b, \quad \text{and}\quad \etr=\eun\times\ede = \mathbf b/b.}
The different components of any vector $\mathbf V$ will be denoted as
\e{\label{eq:covcoord}
\mathbf V = V_\psi \Vpsi + V_\theta \Vthe + V_\phi \Vphi= V_1 \eun + V_2 \ede + V_3 \etr.}

Because the covariant, contravariant and orthonormal bases depend on the space variables, any derivative of a vector quantity expressed in any of these bases will involve a derivative of the basis vectors, which becomes a remainder of order $\epsilon$ thanks to the scaling in the mode conversion region. For example one has 
\begin{equation*}
\epsilon \dpsi \left(V \Vthe\right) =\epsilon \dpsi V  \Vthe + V \epsilon \dpsi \Vthe  = \epsilon \dpsi V  \Vthe + O(\epsilon),
\end{equation*}
so that in the covariant basis the divergence operator reads
\e{\label{eq:orthdiv}\tab{rl}{
\nabla\cdot \mathbf V 
&\displaystyle 
= \epsilon\left( \Vpsi\dpsi + \Vthe\dthe+\Vphi \dphi \right)\cdot\left(\Vpsi V_\psi + \Vthe V_\theta+\Vphi V_\phi  \right),\\
& = \mathfrak D_1V_1+\mathfrak D_2V_2+\mathfrak D_3V_3+  O(\epsilon),
}}
where the differential operators $\mathfrak D$ are defined by
\enn{\mathfrak D_1 = \eun\cdot \nabla = \epsilon \eun\cdot \left( \Vpsi \dpsi + \Vthe \dthe \right),
\text{   }
\mathfrak D_3 = \etr\cdot \nabla =  \frac{b_t}{ \epsilon r Q b}\epsilon (\dthe +Q  \dphi),}
\enn{\mathfrak D_2 =  \ede\cdot \nabla = \left(bQ\left|\Vpsi\right|\right)^{-1}\epsilon \left( -b_t^2\dthe + Qb_p^2 \dphi \right).}

Since the equilibrium is assumed to be axisymmetric, the fields vary as $e^{\imath N\phi}$, where $N$ satisfies $\epsilon N = O(1)$. This essentially leads to a 2D reduction of the 3D model by restricting the study to the poloidal plane. In this context the mode conversion region then becomes a point, as the intersection of two cut-off curves, and lies at the origin of the poloidal plane $( \psi,\theta)=(0,0)$.

\subsection{The cold plasma model}

We consider here a toroidally confined plasma. In an axisymmetric equilibrium state we study an incoming wave propagating in the poloidal plane as a linear perturbation, reducing the model to a 2D problem. Different mode conversion processes exist, and this work focuses on mode conversion between the so-called ordinary (O) and extraordinary (X) propagation modes. These propagation modes can be described in terms of components of the wave electric field with respect to the direction of the confining magnetic field : a pure O-mode wave electric field only has a parallel component, while a pure X-mode wave electric field only has perpendicular components. The X-mode wave considered in this work is a left-handed polarized wave.

The cold plasma model corresponds to propagation of time harmonic electromagnetic waves through zero-temperature plasma. Maxwell's equations are combined with a linearized momentum equation for the particle motion in a stationary confining magnetic field $\mathbf b$. The thermal speed is neglected with respect to the wave speed. Even though this work is focused on high frequency waves, this model encompasses a much broader range of wave motion than magneto-hydrodynamic models. The coupling between the electro-magnetic fields and the fluid motion appears via the current generated by particle motion, modeled as a source term in Maxwell's equations. 
Frequencies will be expressed in $\omega$ units while distances will be expressed in  $c/\omega$ units, where $c$ stands for the speed of light in vacuum. The resulting time-harmonic system reads 
\begin{equation}\label{eq:FL}
\nabla \times \mathbf E = -\imath \mathbf B,
\end{equation}
\begin{equation}\label{eq:AL}
 \nabla \times \mathbf B = \imath \kappa \mathbf E,
\end{equation}
where $\kappa$ is the dielectric tensor for the cold plasma model. It is classically expressed in a right handed orthonormal basis whose third vector is aligned with the magnetic field $\mathbf b$, so in particular in the $(\eun,\ede,\etr)$ basis:
\e{\label{eq:DE}\kappa = \begin{pmatrix}
\kappa_\perp &-\imath \kappa_\wedge & 0 \\
\imath \kappa_\wedge & \kappa_\perp & 0 \\
0&0&\kappa_{\parallel},
\end{pmatrix}} 
where all the coefficients are varying in space.

Following the analysis of the dispersion relation for the cold plasma model, we can describe different cut-offs as surfaces between propagative and evanescent zones for a given type of wave. The corresponding type of wave, impinging from the propagative zone on a cut-off, would then be reflected by the cut-off toward the propagative zone. The O-mode cut-off is the surface defined by the condition $\kappa_{\parallel} = 0$ while the left X-mode cut-off is the surface defined by the condition $\kappa_\perp - \kappa_\wedge = 0$. A pure O-mode wave can only propagate if $\kappa_{\parallel}>0$ while a pure left-handed polarized X-mode wave can only propagate if $\kappa_\perp - \kappa_\wedge>0$. The mode conversion occurs in a neighborhood of the intersecting curve of these two surfaces.
The model derived in this work was developed in   \cite{HaroldConv}, and relies on an expansion in the mode conversion region. The resulting equation, namely \eqref{eq:2ndOFnt}, inherits from the cold plasma model the fact that it has variable coefficients. 

\subsection{A differential system in the mode conversion region}
This paragraph describes the reduction of the $6\times 6$ system  describing Maxwell's equations \eqref{eq:FL} and \eqref{eq:AL} to a $2\times 2$ system, by eliminating some convenient unknowns. 
Moreover a further simplification is performed by specifying the phase of the desired solutions.

The first idea is to eliminate the $(E_1,E_2)$ components. To that purpose,  combining the $\ede$ and $\etr$ components of Faraday's law \eqref{eq:FL}, one can show that these two components satisfy
\enn{
\dpsi (E_1\pm \imath E_2) = (\mathfrak D_1\pm \imath \mathfrak D_2) E_\phi \mp (B_1\pm \imath B_2)\frac{\epsilon r b_t}{b} + \imath B_3\frac{\epsilon r b_p}{b}.
}
As a result, the first two components of 
 \eqref{eq:AL} yield a pair of equations independent of the $E_1$ and $E_2$ components:
\e{\label{eq:AddEq}
\tab{l}{
\displaystyle (\kappa_\perp\mp\kappa_\wedge)\left[ (\mathfrak D_1\pm \imath \mathfrak D_2)E_\phi \mp (B_1\pm \imath  B_2)\frac{\epsilon r b_t}{b} + \imath B_3\frac{\epsilon r b_p}{b} \right] 
\\ \phantom{vvvvvvvvvvvvvvv} = \pm \dpsi \mathfrak D_3(B_1\pm \imath B_2)+\dpsi (\mp\mathfrak D_1 - \imath \mathfrak D_2)B_3 +O(\epsilon).
}}
To complement the latter into the system that will later be expanded in the mode conversion region, express the third components of Faraday's and Ampere's laws in the orthonormal basis, as well as the divergence free condition for the magnetic field:
\begin{equation}
\label{eq:E3}
\dpsi E_3 = -\imath  B_1\frac{\epsilon r b_p}{b} + \mathfrak D_3 E_\phi,
\end{equation}
\begin{equation}
\label{eq:B12}
\imath \kappa_\parallel E_3 = -\mathfrak D_2B_1 +\mathfrak D_1B_2 +O(\epsilon),
\end{equation}
\begin{equation}
\label{eq:divMF}
\mathfrak D_1B_1+\mathfrak D_2B_2+\mathfrak D_3B_3 +  O(\epsilon) = 0.
\end{equation}

Equations \eqref{eq:AddEq}, \eqref{eq:E3}, \eqref{eq:B12}  and \eqref{eq:divMF} form a simplified system in which the toroidal component of the electric field, $E_\phi$, still appears. 
A series of hypotheses will then allow further simplification  of this system, by scanning the relative orders of the different terms to identify the leading order terms.

\begin{hp}\label{hp:OM}
The wave amplitudes of the fields $\mathbf E$ and $\mathbf B$ are expected to vary faster than the equilibrium scale length, but slower than the wave variation $\psi/\epsilon$, $\theta/\epsilon$. More specifically :
\enn{ (\mathbf E,\mathbf B) = (\mathbf E'(\psi\epsilon^{-1/2},\theta\epsilon^{-1/2}),\mathbf B'(\psi\epsilon^{-1/2},\theta\epsilon^{-1/2}))\exp \imath \frac{X(\psi,\theta)+\epsilon N\phi}{\epsilon},}
even if this representation of the magnetic and electric fields is not unique.
\end{hp}

\begin{definition}\label{df:phase}
 The phase will be denoted : $\mathcal X(\psi,\theta,\phi)=\frac{X(\psi,\theta)+\epsilon N\phi}{\epsilon}$. The local wave number is then defined as
\enn{\mathbf k = \epsilon \nabla \mathcal X = \epsilon \nabla \left( \frac{X(\psi,\theta)+\epsilon N\phi}{\epsilon} \right) = \Vpsi \dpsi X + \Vthe \dthe X + \frac{\mathbf e_\phi}{\epsilon r} \epsilon N.}
Taking the derivative of any component of $\mathbf k$ with respect to $\phi$ then reduces to a simple multiplication by $\epsilon N$.
\end{definition}

\begin{hp}\label{hp:WN}
The components of the local wave number that are perpendicular to $\mathbf b$ in the mode conversion region satisfy
\e{\label{eq:OCO}
k_1 = O(\sqrt \epsilon) \text{ and } k_2 = O(\sqrt \epsilon),
}
while the parallel component $k_3 = k_\parallel$ satisfies
\e{\label{eq:XCO}
k_3^2 = 1 - \frac{\omega_{pe}^2}{1 -\Omega_{ce}},}
which corresponds to the usual X-mode cut-off condition 
\begin{equation}\label{eq:Xmco}
\kappa_\perp - \kappa_\wedge- k_3^2 = O(\sqrt \epsilon).
\end{equation}
Moreover the phase $\mathcal X$ is of order $1$.
\end{hp}\label{hp:cp}
\begin{hp}\label{hp:BsE}
The components $E_\phi$, $B_1$ and $B_2$ scale as $O(1)$.
\end{hp}

Thanks to the order of the two first components of the wave number, see Equation \eqref{eq:OCO} from Hypothesis \ref{hp:WN}, then $k_3 = \frac{\epsilon N b}{\epsilon r b_t} + O(\sqrt{\epsilon})$. As a result $\mathfrak D_1$ and $\mathfrak D_2$ are operators of order $O(\sqrt \epsilon)$, while $\mathfrak D_3$ is  of order $O(1)$. Then  since Equations \eqref{eq:divMF} and \eqref{eq:AddEq} show that $(\kappa_\perp \mp \kappa_\wedge - k_3^2)(B_1\pm \imath B_2) = O(\sqrt \epsilon)$, Equation \eqref{eq:XCO} from Hypothesis \ref{hp:WN} provides a way to eliminate the $B_2$ component since:
\enn{ B_2= -\imath B_1 +  O(\sqrt \epsilon).}

Two more unknowns can be eliminated thanks to Equations \eqref{eq:E3} and \eqref{eq:divMF}:
\enn{k_3 E_\phi = \epsilon N E_3 + \frac{\epsilon r b_p}{b}B_1 + O(\sqrt \epsilon),}
\enn{\imath k_3 B_3 = -\mathfrak D_2B_2 - \mathfrak D_1 B_1 + O(\epsilon).}

At this point, $B_2$, $E_\phi$ and $B_3 $ are expressed explicitly in terms of the two last unknowns, namely $E_3$ and $B_1$, while $E_3$ and $B_1$ satisfy an $2\times 2 $ differential system independent of the other unknowns. A new unknown is finally defined in order to obtain a diagonal differential operator :
\begin{definition}
Define the new variable $\tilde{E_3} = E_3 + 2\epsilon r b_p / (b\epsilon N) B_1$.
\end{definition}
As a result, we now get the following system:
\syst{l}{\label{sys:fexp}
\displaystyle (\mathfrak D_1 - \imath \mathfrak D_2)B_1  =-\kappa_\parallel \tilde{E_3} + 2\kappa_\parallel\frac{\epsilon r b_p}{b\epsilon N} B_1 + O(\epsilon),\\
\displaystyle  (\mathfrak D_1 + \imath \mathfrak D_2)\tilde{E_3} = -2\kappa_\parallel \frac{\epsilon r b_p}{b\epsilon N} \tilde{E_3} \\\displaystyle 
\phantom{ (\mathfrak D_1 + i\mathfrak D_2)\tilde{E_3} =}
 +\left\{ 2\kappa_\parallel \left(\frac{\epsilon r b_p}{b\epsilon N}\right)^2 + (\kappa_\perp - \kappa_\wedge-k_3^2) \frac{\epsilon r b_t}{k_3b\epsilon N}
\right\}2B_1  + O(\epsilon).
}

The key point of the next step of the simplification process is to specify the phase function of the electric and magnetic fields introduced in Definition \ref{df:phase}, in order to cancel some of the differential operators terms and obtain a simpler system on the amplitude of the unknowns $\tilde E_3$ and $B_1$. Notice that specifying the phase function actually means looking for only a class of solutions to System \eqref{sys:fexp}.

In order to simplify the $\mathfrak D_2$ term, for any component $A$ of the electric or magnetic field, $A'$ standing for the amplitude of $A$. From the definition of the differential operator, one has
 \enn{\mathfrak D_2 A = \frac{b_t^2}{Qb\left| \Vpsi \right| } \left( -\left({\sqrt \epsilon} \dthe A' + \imath \partial_\theta X A'\right) +\frac{Qb_p^2}{b_t^2} \imath \epsilon N A'\right)e^{\mathcal X} ,}
 so that, choosing $\partial_\theta X =  \left.\frac{\epsilon NQ b_p^2}{b_t^2}\right|_{(0,0)}$, it becomes $\mathfrak D_2 A =- \frac{b_t^2}{Qb\left| \Vpsi \right| }{\sqrt \epsilon} \dthe A'e^{\mathcal X}+O\left(\sqrt{\epsilon}\right)$. Indeed the scaling introduced in Definition \ref{df:phase} ensures that $\sqrt{\epsilon}\dthe A' =O(1)$ while the remaining terms scale as $O((\psi,\theta)) = o(1)$. Following the same idea, the equivalent simplification of $\mathfrak D_1$ is obtained by setting $\eun\cdot \Vpsi \dpsi X + \eun\cdot \Vthe \dthe X = 0$. Choosing the corresponding value of $\dpsi X$, it yields $\mathfrak D_1 A = \sqrt \epsilon \eun\cdot\left( \Vpsi \dpsi A' + \Vthe \dthe A' \right)e^{\mathcal X}+o(1)$. To summarize:

\begin{definition}
From now on, the solutions of \eqref{sys:fexp} are more specifically sought with 
\enn{X(\psi,\theta) =X(0,0) +  \left.\frac{\epsilon NQ b_p^2}{b_t^2}\right|_{(0,0)}\theta - \left.\left(\eun\cdot\Vthe\left| \Vpsi \right|^{-1} \frac{\epsilon NQ b_p^2}{b_t^2}\right)\right|_{(0,0)} \psi.}
\end{definition}

At last the definitions of the O and X mode cut-offs in the mode conversion model lead to the final simplification idea:
\begin{definition}  
The O and X-mode cut-offs are respectively defined by 
$$\kappa_\parallel(\psi/ \sqrt \epsilon) = O(\sqrt \epsilon) \text{ and }
\big[\kappa_\perp - \kappa_\wedge- k_3^2\big](\psi/ \sqrt \epsilon,\theta / \sqrt \epsilon) = O(\sqrt \epsilon),$$
see Hypothesis \ref{hp:WN}.
 The mode conversion point is uniquely defined as the point where 
$\kappa_\parallel = \kappa_\perp-\kappa_\wedge-k_3^2=0$.
\end{definition}
 Notice that each of the coefficients appearing in the right hand side of System \eqref{sys:fexp} is proportional to a quantity vanishing at one cut-off: either $\kappa_\parallel$ or $\kappa_\perp - \kappa_\wedge-k_3^2$. A spontaneous change of variables and the subsequent differential operators are then:
\begin{definition}\label{df:cov+cst}
The constant $c_1$ is defined by the local expansion $- 2\kappa_\parallel= \sqrt \epsilon c_1 \psi / \sqrt \epsilon$, and $c_2$ and $c_3$ are defined by $ 2\left(\frac{\epsilon r b_p}{b\epsilon N}\right)^2\kappa_\parallel + \frac{\epsilon r b_t}{b k_3\epsilon N} (\kappa_\perp - \kappa_\wedge-k_3^2) = \sqrt \epsilon \left( c_2 \psi / \sqrt \epsilon +  c_3 \theta / \sqrt \epsilon\right)$. In other words
$$
c_1 = \frac{-2 \dpsi \kappa_\parallel(0)}{\sqrt \epsilon}, \ 
c_2 = \frac{\dpsi F(0,0)}{\sqrt \epsilon}, \ 
c_3 = \frac{\dthe F(0,0)}{\sqrt \epsilon} ,
$$
where $F(\frac{\psi}{\sqrt\epsilon},\frac{\theta}{\sqrt\epsilon}) = 2\left(\frac{\epsilon r b_p}{b\epsilon N}\right)^2\kappa_\parallel + \frac{\epsilon r b_t}{b k_3\epsilon N}(\kappa_\perp - \kappa_\wedge-k_3^2) $.

Define the change of variables
\sysnn{l}{
y' = (\theta + c_2 \psi / c_3 ) \epsilon^{-1/2}, \\
x' = \psi \epsilon^{-1/2},
}
 while $ \alpha = \left( \frac{\epsilon r b_p}{b\epsilon N} \right)_{(0,0)}$ and the constants are $d_1'$, $d_2'$ and $d_3'$ defined by the related operators
\sysnn{rl}{
\mathfrak D_1' &\displaystyle =  \eun\cdot \left( \Vpsi \partial_{x'} + \left(\Vpsi \frac{c_2}{c_3} + \Vthe\right) \partial_{y'}  \right) = d_1'\partial_{x'} + d_2' \partial_{y'} , \\
\mathfrak D_2' &\displaystyle = - \left( \frac{b_t^2}{Qb}\left|\Vpsi\right|^{-1}  \right)_{(0,0)} \partial_{y'}  = d'_3 \partial_{y'}.
}
\end{definition}

Thanks to these definitions, the leading order terms of System \eqref{sys:fexp} give the following system
\sysnn{rl}{
 (\mathfrak D'_1 + \imath \mathfrak D'_2)\tilde{E_3'} &\displaystyle =  c_1 x' \alpha \tilde{E_3'} + c_3 y' (2B_1'),\\
(\mathfrak D'_1 - \imath \mathfrak D'_2)(2B_1') &\displaystyle =  c_1 x' \tilde{E_3'} -  c_1 x'\alpha (2B_1'),
}
where it is crucial that the unknowns are now the amplitudes of the original unknowns.
For the sake of clarity, another rescaling change of variables is introduced by the following definition.
\begin{definition}\label{df:Fcov+cst}
Define $\tilde{e}_1$ and $\tilde{e}_2$ such that $\tilde{e}_1  = \sqrt{\frac{d_1'}{\alpha c_1}}$ and $\tilde{e}_2 = \frac{\alpha d_1'}{c_3\tilde{e}_1}$.
Define a scaling of variables and unknowns by
\sysnn{l}{
x' = \tilde{e}_1 x \text{ and } y' = \tilde{e}_2 y, \\
2\alpha B_1' = \mathsf B \text{ and } \tilde E_3' = \mathsf E,
}
together with 
$d_2 = \frac{d'_2\tilde{e}_1}{d'_1\tilde{e}_2}$, $d_3 =  \frac{d'_3\tilde{e}_1}{d'_1\tilde{e}_2}$ and $d = d_2+\imath d_3$.
\end{definition}

The resulting system finally reads:
\begin{equation}\label{sys:DDbar}
\left\{\tab{rl}{
(\partial_x +d\partial_y) \mathsf E &=x\mathsf E+y\mathsf B, \\
(\partial_x +\overline d\partial_y) \mathsf B &=x\mathsf E-x\mathsf B.
}\right.
\end{equation}
It is a $2\times 2$ differential system model in the mode conversion region. An important feature of this system is that, in the $(x,y)$ coordinates, the cut-offs are defined by $x=0$ and $x+y=0$. It is a direct consequence of the fact that in the intermediate variables they respectively correspond to $x'=0$ and $c_1\alpha^2 x'+c_3y' = 0$,  combined with the final rescaling. As to the mode conversion region, it is the neighborhood of the point $\{x=0\}\cap \{x+y=0\}$, that is to say the origin $(x,y)=(0,0)$.

Since from now on the only remaining parameter is $d\in\mathbb C$, define $d_i = \Im (d)$ and $d_r = \Re (d)$.

\subsection{A second order equation}
In order to write System \eqref{sys:DDbar} as a single second order equation, one would naturally eliminate one of the components and obtain one of the following equations:
\begin{equation*}
(\partial_x +\overline d\partial_y + x )
  \left(\frac1 y(\partial_x +d\partial_y) \mathsf E \right)
=(\partial_x +\overline d\partial_y + x )
  \left(\frac x y\mathsf E\right)+x\mathsf E,
\end{equation*}
\begin{equation*}
(\partial_x +d\partial_y - x )
  \left(\frac1 x(\partial_x +\overline d\partial_y) \mathsf B \right)
=y\mathsf B -(\partial_x +d\partial_y - x )
  \mathsf B.
\end{equation*}
However both of these equations have an artificial singularity at the mode conversion point. Obtaining a well-behaved second order equation for System \eqref{sys:DDbar} then requires further investigation. 

Looking for a physical solution, one can perform a change of variables involving the exponential of a quadratic form, as proposed in   \cite{HaroldConv} (Equation (54)). Suppose that a quadratic form $\mathsf Q(x,y) = \frac 1 2 (K x^2 + 2L xy + M y^2)$ satisfies $(\mathsf E,\mathsf B) = (\tilde{\mathsf E},\tilde{\mathsf B}) e^{\mathsf  Q(x,y)}$. The idea is to determine $\mathsf Q$ such as to simplify the differential system satisfied by the amplitude functions $(\tilde{\mathsf E},\tilde{\mathsf B})$ and such that $\Re \mathsf Q(x,y) \leq 0$ to guarantee a physical behavior at infinity. 

Starting from System \eqref{sys:DDbar}, it is straightforward to see that the amplitudes $(\tilde{ \mathsf E},\tilde{\mathsf B})$ satisfy
\begin{equation}\label{sys:DDbartilde}
\left\{\tab{rl}{
(\partial_x +d\partial_y) \tilde{\mathsf E} &=((1-K-dL)x-(L+dM)y)\tilde{\mathsf E}+y\tilde{\mathsf B}, \\
(\partial_x +\overline d\partial_y) \tilde{\mathsf B} &=x\tilde{\mathsf E}-((1+K+\overline d L)x+ (L+\overline d M)y)\tilde{\mathsf B}.
}\right.
\end{equation}
A simplification of the right hand side is then performed thanks to an adequate set of constants $(K,L,M)$, setting
\begin{equation}\label{eq:simpl0}
L+dM = 1+K+\overline d L = 0,
\end{equation}
\begin{equation}\label{eq:simplmu}
1/\mu = K+dL-1 = 1/(L+\overline d M),
\end{equation}
so that System \eqref{sys:DDbartilde} becomes
\begin{equation}\label{sys:EBmu}
\left\{\tab{rl}{
(\partial_x +d\partial_y) (\mu\tilde{\mathsf E}) &=- x\tilde{\mathsf E}+\mu y\tilde{\mathsf B}, \\
(\partial_x +\overline d\partial_y) \tilde{\mathsf B} &=x\tilde{\mathsf E}- \mu y\tilde{\mathsf B}.
}\right.
\end{equation} 
From equations \eqref{eq:simpl0} and \eqref{eq:simplmu}, it is straightforward that $d\mu = 2+1/\mu$,
\begin{equation}\label{eq:KLM}
K = \left(-1+\frac{d \overline{d}\mu}{\overline{d}-d}\right), \
L =  - \frac{d\mu}{\overline{d}-d} \text{ and }
M = \frac{\mu}{\overline{d}-d}.
\end{equation}
So the quadratic form reads
$ 
\mathsf Q (x,y) = \frac 1 2\left(\left(-1+\frac{d \overline{d}\mu}{\overline{d}-d}\right) x^2 - 2\frac{d\mu}{\overline{d}-d}xy + \frac{\mu}{\overline{d}-d}y^2\right)$. The physical behavior of the solution at infinity is linked to the sign of $\Re \mathsf Q$,  since if $\Re \mathsf Q \leq 0$  there is no propagation for $|(x,y)|\rightarrow \infty$. 
A little more algebra then leads to
\begin{equation*}
2\Re \mathsf Q (x,y)   =  -\frac{\mu_i}{2d_i} \left(y+\frac{x}{|\mu|^2} \right)^2  =  \frac{1}{2d_i|\mu|^2}\Im \frac 1 \mu \left(|\mu|^2y+x \right)^2.
\end{equation*}
The parameter $\mu$ satisfies the second degree polynomial equation $d\mu^2-2\mu-1 = 0$. The two roots of the polynomial $X^2+2X-d$ are $(1/\mu)_\pm= -1\pm \sqrt{1+d}$, and the inequality $\Re \mathsf Q \leq 0$ is satisfied if and only if $d_i$ and $\Im \frac 1 \mu$ have opposite signs. Since the sign of $\Im \sqrt{1+d}$ is the sign of $d_i$, it is then clear that the root $(1/\mu)_- = -1- \sqrt{1+d}$ ensures the desired estimate at infinity for $\Re \mathsf Q$.
The following lemma summarizes this result.

\begin{lemma}
Given the complex numbers $(K_\pm,L_\pm,M_\pm)$ satisfying  \eqref{eq:KLM} with $1/\mu$ computed as $(1/\mu)_\pm$, and the corresponding quadratic forms $\mathsf Q_\pm(x,y) = \frac 1 2 (K_\pm x^2 + 2L_\pm xy + M_\pm y^2)$. The form $\mathsf Q_-$ corresponding to the root $(1/\mu)_- = -1-\sqrt{1+d}$ satisfies 
\begin{equation}\label{eq:ReQineg}
\Re \mathsf Q_- (x,y) \leq 0, \quad \forall (x,y)\in \mathbb R^2.
\end{equation}
As a result the function $ \exp \mathsf Q_-(x,y)$ is bounded.
\end{lemma}
So define now $(\tilde{\mathsf  E},\tilde{ \mathsf B})=(\mathsf  E,\mathsf  B) \exp (-\mathsf Q_-(x,y))$. Since the two right hand sides of System \eqref{sys:EBmu} are the same up to a multiplicative constant, a divergence free condition holds :
$
\partial_x(\mu \tilde{\mathsf  E}+\tilde{\mathsf  B}) + \partial_y(\mu d \tilde E + \overline d  \tilde{\mathsf  B}) = 0
$. As a result there is a potential $\varphi$ such that $-\partial_y \varphi = \mu \tilde{\mathsf  E}+\tilde{\mathsf  B}$ and $\partial_ x \varphi  = \mu d \tilde{\mathsf  E} + \overline d \tilde{\mathsf  B}$. 
For the sake of simplicity, define the differential operator $\mathfrak D  =\partial_x +d\partial_y $. Then $\tilde F = \frac{\varphi}{\mu(d-\overline d)}$ satisfies $\tilde{\mathsf  E}  = \overline{\mathfrak D}\tilde F$, $\tilde{\mathsf  B}  = -\mu\mathfrak  D\tilde F$ and the second order differential equation 
\begin{equation}
\label{eq:2ndOF}
\mathfrak D\overline{ \mathfrak D} \tilde F +\frac{x}{\mu}\overline{\mathfrak  D} \tilde F + \mu y \mathfrak D \tilde F = 0.
\end{equation}
At this point this work finally diverges from   \cite{HaroldConv}. Since we wanted a well-behaved second order equation for System \eqref{sys:DDbar}, we can undo the $e^{\mathsf Q}$ transformation, to get an equation for $F = \tilde F \exp \mathsf Q_-(x,y)$. 
From \eqref{eq:KLM} stems that $\mathfrak D\mathsf Q = \left(1+\frac{1}{\mu}\right)x$, $\overline{ \mathfrak D} \mathsf Q = \mu y -x$ and $\mathfrak D\overline{ \mathfrak D }\mathsf Q = d\mu-1 = 1+\frac{1}{\mu}$. The definition of $F$ implies that
$\mathfrak D\tilde F = (\mathfrak DF-F\mathfrak D\mathsf Q)e^{-\mathsf Q}$ and 
$\overline{\mathfrak  D}\tilde F = (\overline{\mathfrak  D}F-F\overline{ \mathfrak D}\mathsf Q)e^{-\mathsf Q}$. Then \eqref{eq:2ndOF} directly gives
\begin{equation*}
\mathfrak  D\overline{ \mathfrak D} F + 2\imath d_ix\partial_y F-\left(1+\frac{1}{\mu}+x(x+y)\right)F = 0,
\end{equation*} 
which is precisely Equation \eqref{eq:2ndOFnt}.

\section{Theoretical study}
\label{sec:TS}
In order to work on a well-posed problem, a weak formulation will be derived focusing on coercivity and symmetry properties.  Remember that $d=d_r+\imath d_i$ is a complex constant and that $1/\mu = -1-\sqrt{1+d}$.

Consider the second order term of the elliptic equation  \eqref{eq:2ndOFnt}. The corresponding differential operator can be written in a divergence form, as  
\begin{equation}\label{eq:DDbarM}
\mathfrak  D\overline{ \mathfrak D} = (\dx + d \dy)(\dx+\overline d \dy) 
= \dx^2 + (d+\overline d)\dx\dy +|d|^2 \dy^2 
= \nabla\cdot\left(\begin{pmatrix}
1 & d \\ \overline d & |d|^2
\end{pmatrix}\nabla \right).
\end{equation}
Since the eigenvalues of the matrix $M = \begin{pmatrix}
1 & d \\ \overline d & |d|^2
\end{pmatrix}$ are $0$ and $1+|d|^2$, the lower bound estimate associated with the form \eqref{eq:DDbarM} reads: $\forall X \in\mathbb C^2$,
$
MX\cdot X \geq \min(0,1+|d|^2) \|X\|^2
$. 
But because $\min(0,1+|d|^2) =0$ this is not an ellipticity condition for the  $\mathfrak  D\overline{ \mathfrak D}$ operator.
 However this operator can also be written 
\begin{equation}\label{eq:DDbarA}
 \mathfrak  D\overline{ \mathfrak D}= (\dx + d \dy)(\dx+\overline d \dy) 
= \dx^2 + 2d_r\dx\dy +|d|^2 \dy^2
  = \nabla\cdot\left(\begin{pmatrix}
1 & d_r \\ d_r & |d|^2
\end{pmatrix}\nabla \right).
\end{equation}
Since the eigenvalues of the matrix $A = \begin{pmatrix} 1 & d_r \\ d_r & |d|^2\end{pmatrix}$ are  
$$\lambda_\pm=\frac{1+|d|^2\pm\sqrt{(1-|d|^2)^2+4d_r^2}}{2} = \frac{1+|d|^2\pm\sqrt{(1+|d|^2)^2-4d_i^2}}{2} ,$$
 the lower bound estimate associated with the form \eqref{eq:DDbarA} reads: $\forall X \in\mathbb C^2$,
$
AX\cdot X \geq \lambda_- \|X\|^2
$. 
Consequently it is sufficient to suppose that $d_i\neq 0$ for the latter to be an ellipticity condition for the $\mathfrak  D\overline{ \mathfrak D}$ operator.
So the weak formulation will be based on the identity  $\mathfrak  D\overline{ \mathfrak D} u = \nabla\cdot(A\nabla u) $. 

As to the first order term of \eqref{eq:2ndOFnt} it will be split in the weak formulation in a symmetric way, thanks to the fact that it can be written as  $2\imath d_i \mathbf v \cdot \nabla F$ with $\mathbf v = \begin{pmatrix} 0\\x\end{pmatrix}$. Indeed, this vector field $\mathbf v$ is such that the $x$ component does not depend on the $x$ variable and the $y$ component does not depend on the $y$ variable. So this first order term can be symmetrized as $\int_\Omega (\mathbf v \cdot \nabla F) \overline \varphi = \frac 12 \int_\Omega (\mathbf v \cdot \nabla F)\overline \varphi + \frac 12 \int_\Omega F  (\mathbf v \cdot \nabla\overline \varphi)$. 
Moreover, this justifies the choice of Equation \eqref{eq:2ndOFnt} over \eqref{eq:2ndOF} since in the latter the first order terms
$
\int_\Omega \frac x \mu \overline{\mathfrak D} \tilde F \cdot \overline{\varphi}$ and 
$
\int_\Omega  \mu y {\mathfrak D} \tilde F \cdot \overline{\varphi}$ do not have such a simple symmetric formulation.
According to these ideas, complement Equation \eqref{eq:2ndOFnt} with an appropriate boundary condition:
\begin{equation}\label{sys:THEsys}
\left\{\begin{array}{ll}
\nabla\cdot (A\nabla u) + 2\imath d_ix\partial_y u-\left(1+\frac{1}{\mu}+x(x+y)\right)u = 0,& (\Omega),\\
 \nu \cdot (A\nabla u) + \imath d_ix\nu_y u  +\imath \sigma u = g,& (\partial\Omega),
\end{array} \right.
\end{equation}
$\Omega$ being a Lipschitz domain around the mode conversion point, i.e. $(0,0)\in\Omega$, $\sigma$ being  positive constant, and $1/\mu = -1-\sqrt{1+d}$.

\begin{definition}
Consider a Lipschitz domain $\Omega$ around the mode conversion point, i.e. $(0,0)\in\Omega$.
A weak solution of the partial differential system \eqref{sys:THEsys} is: for all $v\in H^1(\Omega)$
\begin{equation}\label{eq:WF}
\begin{array}{l}
\displaystyle \int_\Omega A\nabla u 
\cdot \nabla \overline v 
-\imath d_i\int_\Omega x\left( \partial_y u  \overline v - u  \partial_y\overline v \right)
+\int_\Omega c(x,y) u  \overline v
+\imath \sigma \int_{\partial \Omega}
u\overline v
= \int_{\partial \Omega} g \overline v ,
\end{array}
\end{equation}
where $u\in H^1(\Omega)$ and $c(x,y) = 1+\frac 1\mu+x(x+y)$. 
\end{definition}

Under a single hypothesis on the imaginary part of the parameter $d$, the following theorem states the well-posedness of \eqref{eq:WF}.
\begin{theorem}
Consider a Lipschitz domain $\Omega$.
Suppose $d_i < 0$ and $\sigma>0$ ; then there exists a unique solution $u\in H^1(\Omega)$ to the weak formulation \eqref{eq:WF}.
\end{theorem}

\begin{proof} The idea is to first deal with the first order term and then treat the whole problem as a coercive plus compact decomposition.

Define the intermediate problem 
\begin{equation}\label{eq:interWF}
\begin{array}{l}
\displaystyle \int_\Omega A\nabla u 
\cdot \nabla \overline v 
-\imath d_i\int_\Omega x\left( \partial_y u  \overline v - u  \partial_y\overline v \right)
+\lambda\int_\Omega u  \overline v
+\imath \sigma \int_{\partial \Omega}
u\overline v \\\displaystyle 
\phantom{\displaystyle \int_\Omega A\nabla u 
\cdot \nabla \overline v 
-\imath d_i\int_\Omega x\left( \partial_y u  \overline v - u  \partial_y\overline v \right)
+\lambda\int_\Omega u  \overline v}
= \int_{\partial \Omega} g \overline v + \int_{\Omega} f \overline v ,
\end{array}
\end{equation}
where the parameter $\lambda$ is to be tuned to ensure the well-posedness of this problem. Define
\begin{equation}
a(u,v) = \int_\Omega A\nabla u 
\cdot \nabla \overline v 
-\imath d_i\int_\Omega x\left( \partial_y u  \overline v - u  \partial_y\overline v \right)
+\lambda\int_\Omega  u  \overline v
+\imath \sigma \int_{\partial \Omega}
u\overline v.
\end{equation}
It is then classical to write
\begin{equation}
\begin{array}{rl}
\Re a(u,u) &\displaystyle = \int_\Omega A\nabla u 
\cdot \nabla \overline u 
+d_i\int_\Omega x\Im \left( \partial_y u  \overline u \right)
+\lambda \int_\Omega |u|^2,\\
&\displaystyle  \geq \lambda_- \|\nabla u \|^2 - |d_i| \max_{(x,y)\in\Omega} |x|\|\nabla u \|\| u \| + \lambda\| u \|^2,
\end{array}
\end{equation}
and setting $\lambda=\frac{\left(|d_i|\max_{(x,y)\in\Omega} |x|\right)^2 }{2\lambda_-} > \frac{\left(|d_i|\max_{(x,y)\in\Omega} |x|\right)^2 }{4\lambda_-}$ yields
\begin{itemize}
\item considering  $\lambda_-\left(\frac{\|\nabla u\|}{\|u\|}\right)^2-|d_i|\max_{(x,y)\in\Omega} |x|\frac{\|\nabla u\|}{\|u\|} + \lambda$,
 
then $\Re a(u,u)\geq \frac{\lambda}{2}\| u \|^2$,
\item considering  $\lambda \left(\frac{\|u\|}{\|\nabla u\|}\right)^2-|d_i|\max_{(x,y)\in\Omega} |x|\frac{\|u\|}{\|\nabla u\|} + \lambda_-$,

then $\Re a(u,u)\geq \frac{\lambda_-}{2}\| \nabla u \|^2$.
\end{itemize} 
As a result $\Re a(u,u)\geq \frac{\min(\lambda_-,\lambda)}{4}\| u \|_{H^1(\Omega)}^2$, and so there exists a unique solution to \eqref{eq:interWF} thanks to Lax-Milgram theorem. Define $T:(f,g)\in L^2(\Omega)\times L^2(\partial \Omega) \mapsto u\in L^2(\Omega) $, $u$ being the unique solution to \eqref{eq:interWF}. The operator $T$ is compact since the unique solution $u$ of \eqref{eq:interWF} is actually in $H^1(\Omega)$ and the embedding $H^1(\Omega) \hookrightarrow	 L^2(\Omega)$ is compact.

Then $u$ is solution to the initial problem \eqref{eq:WF} if and only if $u$ satisfies $u = T((\lambda+c(x,y))u,g)$, which is equivalent to $u$ satisfying $(I-T((\lambda+c(x,y))\cdot ,0))u = T(0,g)$. Note that the composition function $T(\cdot,0)$ composed with the multiplication by the smooth function $c(x,y)+\lambda$ is compact, as the composition of a smooth and a continuous functions. Since this is a coercive plus compact decomposition, the Fredholm alternative states the equivalence between 
existence and uniqueness of a solution. To prove the uniqueness start from a solution $u$ to the homogeneous equation \eqref{eq:WF}, i.e. with $g=0$. Then the imaginary part of \eqref{eq:WF} reads :
\begin{equation}
\Im\frac 1 \mu \int_\Omega |u|^2 + \sigma \int_{\partial\Omega} |u|^2 = 0.
\end{equation} 
Since $\sigma>0$, $d_i \cdot \Im \frac 1 \mu <0$ and $d_i<0$, it implies that $u=0$. So the solution is unique.
\end{proof}
Note that $d_i\neq 0$ is crucial to prove the coercivity of the bilinear form, while the sign of $d_i$ is crucial to prove the uniqueness.

\section{A numerical method}\label{apps}
The main feature of System \eqref{sys:THEsys} is that two of the coefficients in the equation depend smoothly on the space variables. Moreover the mode conversion  point itself has been defined as the point satisfying both the X- and O-mode cut-off conditions, namely $x+y=0$ and $x=0$. 
Since each cut-off is defined as the level curve of a smooth function, it is then clear that the varying nature of these quantities is crucial to the mode conversion phenomenon.  
As a result it is important to use a numerical method adapted to variable coefficients.

In order to reduce the pollution effect, documented in   \cite{poll}, appearing in finite elements methods for wave propagation, plane wave methods use basis functions adapted to this particular application: these basis functions are solutions of the homogeneous equation, \cite{GHP09_2}. See \cite{trefftz} for the first description of these methods under the denomination  Trefftz-based methods, and   \cite{GHP09_2,plu2} for more recent developments. The leading idea is that basis functions embedding information about the problem of interest are more efficient than polynomial functions to approximate a wave.

The numerical method that we propose in this work relies on  basis functions designed to fit the variable coefficients, called Generalized Plane Waves (GPWs), coupled with a Discontinuous Galerkin (DG) method, called the Ultra-Weak Variational Formulation (UWVF).
The UWVF, proposed by B. Despr\'es in   \cite{d94}, has been recast as a DG method in   \cite{GHP09_2}. It is a Trefftz method and its main feature is that all integrals involved in the formulation are boundary integrals, and as such it is much cheaper to evaluate than volume integrals from other methods. The idea to couple it with GPWs was proposed in   \cite{LMIGBD}.  
 
\subsection{The Ultra-Weak Variational Formulation}
In order to introduce the UWVF, define the slightly more general problem
\begin{equation}\label{sys:gen}
\left\{\begin{array}{ll}
\nabla\cdot(A\nabla u) + 2\imath d_ix\partial_y u-\left(1+\frac{1}{\mu}+x(x+y)\right)u = 0,&(\Omega)\\
 \nu \cdot (A\nabla u) + \imath d_ix\nu_y u +\imath \sigma u = \mathcal Q\left(- \nu \cdot (A\nabla u) - \imath d_ix\nu_y u +\imath \sigma u  \right) + g,&(\partial \Omega)
\end{array} \right.
\end{equation}
where $d=d_r+\imath d_i$ is a real constant, $A = \begin{pmatrix} 1 & d_r \\ d_r & |d|^2\end{pmatrix}$ is therefore a constant matrix, the complex constant $\mu$ was defined by $\frac1\mu = -1-\sqrt{1+d}$,  $\sigma>0$ is a real constant, $\mathcal Q$ a real valued piecewise constant function on the boundary, and $\Omega$ is a Lipschitz domain. 

The UWVF is a weak formulation that relies on the use of test functions satisfying the dual equation. 
Let $v$ be a smooth test function and $u$ be a smooth solution of \eqref{sys:gen} on an open set $\mathcal O$. Then the integration by parts leading to the classical weak formulation of \eqref{sys:gen} yields 
\begin{equation}\label{eq:ibpu}
\begin{array}{l}
\displaystyle \int_{\mathcal O} A\nabla u 
\cdot \nabla \overline v 
-\imath d_i\int_{\mathcal O} x\left( \partial_y u  \overline v - u  \partial_y\overline v \right)
+\int_{\mathcal O} c(x,y) u  \overline v \\\displaystyle
\phantom{\int_{\mathcal O} A\nabla u 
\cdot \nabla \overline v 
-\imath d_i\int_{\mathcal O} x\left( \partial_y u  \overline v - u  \partial_y\overline v \right)
}
-\int_{\partial {\mathcal O}}
\left( \nu \cdot (A\nabla u) + \imath d_ix\nu_y u\right) \overline v
= 0 .
\end{array}
\end{equation}
 Now suppose a smooth test function $v$ satisfies the dual equation 
\begin{equation}\label{eq:L*v}
-\nabla\cdot (A\nabla v) - 2\imath d_ix\partial_y v+\overline{c(x,y)}v = 0.
\end{equation}
The equivalent integration by parts yields
\begin{equation}\label{eq:ibpv}
\begin{array}{l}
\displaystyle \int_{\mathcal O} A\nabla   \overline v 
\cdot \nabla u
-\imath d_i\int_{\mathcal O} x\left( \partial_y u  \overline v - u  \partial_y\overline v \right)
+\int_{\mathcal O} c(x,y) u  \overline v\\\displaystyle
\phantom{\int_{\mathcal O} A\nabla u 
\cdot \nabla \overline v 
-\imath d_i\int_{\mathcal O} x\left( \partial_y u  \overline v - u  \partial_y\overline v \right)
}
-\int_{\partial {\mathcal O}}
 \overline{\left( \nu \cdot (A\nabla v) + \imath d_ix\nu_y v\right)} u
= 0 .
\end{array}
\end{equation}
Noticing that the volume integral terms in \eqref{eq:ibpu} and \eqref{eq:ibpv} are identical since $A$ is hermitian, the UWVF contains only boundary integrals, thanks to such integration by parts. In order to write explicitly the UWVF for our problem, the following definitions are required.

Consider a Lipschitz domain $\Omega\subset \mathbb R^2$, with a mesh $\displaystyle\overline \Omega = \bigcup_{k_1}^{N_h} \overline{\Omega_k }$.  The boundary of the domain is $\Gamma = \partial \Omega$.
Let $h_k$ be the diameter of $\Ok$ and $\rho_k$ be the maximum of the diameters of
the spheres inscribed in $\Ok$, and $\mathcal Q_k$ be the value of $\mathcal Q$ on $\Ok$. The domain is supposed to be meshed so that $\mathcal Q$ is constant on each $\Ok$. The mesh is such that
$\exists \sigma$ such that $h_k \leq \sigma \rho_k$. The refinement parameter or mesh size parameter $h$ is then defined by 
$h = \max h_k$. The terminology {\it{regular mesh}} used to describe such a mesh comes from   \cite{ciarlet}.
The interface between two mesh elements $\Ok$ and $\Oj$, oriented from $\Ok$ to $\Oj$, is denoted $\Sigma_{kj}$. The part of the edge of a mesh element $\Ok$ that is part of the boundary of the domain is denoted $\Gk = \Gamma\cap \bOk$.

\begin{definition}
On each element $\Ok$ of the mesh, the differential operator and its dual are defined on each element of the mesh by
\sysnn{rl}{L\varphi_{|\Ok} &= \frac{1}{\imath \sigma}\left(-\nabla\cdot (A\nabla \varphi) - 2\imath d_ix\partial_y \varphi+c(x,y)\varphi\right), \\
L^*\varphi_{|\Ok} &= \frac{-1}{\imath \sigma}\left(-\nabla\cdot (A\nabla \varphi) - 2\imath d_ix\partial_y \varphi+\overline{c(x,y)}\varphi\right),
}
and define
\sysnn{l}{
l_k(\varphi,\psi) = \frac{1}{\imath \sigma}\left(\int_\Ok A\nabla \varphi 
\cdot \nabla \overline \psi 
-\imath d_i\int_\Ok x\left( \partial_y \varphi  \overline \psi - \varphi \partial_y\overline \psi \right)
+\int_\Ok c(x,y) \varphi  \overline \psi\right) \\
b_k\varphi = \frac{-1}{\imath \sigma}\left( \nu \cdot (A\nabla \varphi) + \imath d_ix\nu_y \varphi\right) \\
b_k^* \varphi=\frac{1}{\imath \sigma}\left(  \nu \cdot (A\nabla \varphi) + \imath d_ix\nu_y \varphi\right) \\
c_k \varphi = \varphi}
where  $\sigma$ is given by the boundary condition.
\end{definition}

These definitions of the operators $b_k$ and $c_k$ rely on their behavior along interior edges of the domain:
\begin{equation*}\left\{
\begin{array}{l}
 (b_k v)_{\skj} =  -(b_j v)_{\sjk}\\
 (c_k v)_{\skj} =  \phantom -(c_j v)_{\sjk}
 \end{array}\right.
\end{equation*}

Denote by $(\cdot,\cdot)_k$ the $L^2$ scalar product on $\Ok$, $(\varphi,\psi)_k = \int_{\Ok} \varphi \overline \psi$. Suppose that $u$ is a solution of \eqref{sys:gen} and $v$ is a solution of \eqref{eq:L*v} ; it is now clear that on each element $\Ok$ of the mesh one has
\begin{equation}
\begin{array}{l}
\displaystyle (Lu,v)_k - l_k(u,v) = -\int_\bOk \frac{1}{\imath \sigma}\nu \cdot (A\nabla u) \overline v -\int_\bOk \frac{1}{\imath \sigma} \imath d_ix\nu_y u\overline v = (b_ku,v)_k, \\
\displaystyle (u,L^*v)_k - l_k(u,v) = -\int_\bOk \frac{1}{\imath \sigma}\nu \cdot (A\nabla\overline v) u  +\int_\bOk \frac{1}{\imath \sigma} \imath d_ix\nu_y u\overline v, \\
\displaystyle \phantom{(u,L^*v)_k - l_k(u,v)} = \int_\bOk u\overline{\frac{1}{\imath \sigma}\nu \cdot (A\nabla v)   + \frac{1}{\imath \sigma} \imath d_ix\nu_y v} = (c_k u , b_k^*v)_k,
\end{array}
\end{equation}
 which yields, summing over $k$ and using the boundary condition in \eqref{sys:gen},
\begin{equation}
\begin{array}{l}
\displaystyle \sum_k( -b_k u + c_k u , b_k^* v+c_kv  )_k
\displaystyle - \sum_k   \sum_{j\neq k}( -b_j u + c_j u , -b_k^* v+c_kv  )_k \\
- \sum_{k, \Gamma_k \neq \emptyset}\mathcal Q_k ( -b_k u + c_k u , -b_k^* v+c_kv  )_k 
\displaystyle = \sum_{k, \Gamma_k \neq \emptyset}( g , -b_k^* v+c_kv )_k.
\end{array}
\end{equation}

 The function space for the UWVF is
$
\displaystyle\mathcal  V=\prod_{k \in \unNh} L^2(\bOk),
$
while the test function space is defined by
\begin{equation*}
 \mathcal H=\prod_{k=1}^{N_h} \mathcal H_k \textrm{ where }\mathcal H_k = \left\{ v_k \in H^1(\Ok),\left| \begin{array}{l}
                                        L^* v_k = 0, (\Ok),
					\\ \left( 
b_k^* v_k + c_k v_k \right)_{|\bOk} \in L^2(\bOk)
                                       \end{array}
 \right. \right\}.
\end{equation*}
As a result of these definitions, any element of $\mathcal V$ is actually defined on the edges of every element of the mesh.

\begin{theorem}\label{thuwvf}
 Let $u\in H^1 (\Omega)$ be a solution of problem \eqref{sys:gen} such that $\partial_{\nu_k} u \in L^2(\bOk)$ for any k.
Let $\sigma>0$ be a given real number and $\mathcal Q$ is such that $\mathcal Q|{\bOk}=\mathcal Q_k\in\mathbb R$ and $|\mathcal Q_k|<1$ for all $k$.
 Then $X\in \mathcal V$
 defined by $X_{|\bOk}=X_k$ with
 $X_k=((-b_k + c_k)u_{|\Omega_k})_{|\bOk}$ satisfies 
\begin{equation}\label{UWVF}
\begin{array}{ll}
\displaystyle
 \sum_k \left( \int_{\bOk} \frac{1}{\sigma} X_k
 \overline{(b_k^*+c_k) e_k}
- \sum_{j,j\neq k} \int_{\Sigma_{kj}} \frac{1}{\sigma} X_j
 \overline{(-b_k^*+c_k)e_k} \right) 
\\ 
\displaystyle
 -\sum_{k, \Gamma_k\neq \emptyset} \int_{\Gamma_k} \frac{\mathcal Q_k}{\sigma} X_k \overline{(-b_k^*+c_k)e_k}
= \sum_k \int_{\Gamma_k} \frac{1}{\sigma}g\overline{(-b_k^*+c_k)e_k}, 
\end{array}
\end{equation}
for any 
$
 e= (e_k)_{k\in \unNh}\in \mathcal H$.
 Conversely, if $X\in \mathcal V$ is solution of $\eqref{UWVF}$ then the function $u$ defined locally  as the weak solution of:
\begin{equation} \label{eqAA}
\left\{
\begin{array}{l}
 u_{|\Omega_k}= u_k \in H^1(\Omega_k), \\
 L u_k = f_{|\Omega_k}, \\
(-b_k+c_k) u_k = X_k,
\end{array}
\right.
\end{equation}
is the unique solution of the problem \eqref{sys:gen}.
\end{theorem}
This result is classical in the context of UWVF. We refer
to   \cite{cd98,monk:buffa,hip1,monk2,monk3}. Equation \eqref{UWVF} is the UWVF for \eqref{sys:gen}. Note that if $u$ satisfies $\nabla \cdot(A\nabla u)\in L^2$ on a Lipschitz domain, then the normal derivative of $u$ on the boundary of that domain is well defined. If moreover $u$ satisfies the boundary condition of \eqref{sys:gen}, then the normal derivative of $u$ can be written $\nu\cdot\nabla u= \frac{\mathcal Q-1}{\mathcal Q+1}(\imath \sigma u - \imath d_ix\nu_y u) +\frac{g}{\mathcal Q+1}$, so that it belongs to $L^2$.

\subsection{Discretization}
In order to discretize exactly the UWVF, one would need to use basis functions satisfying the dual equation $L^*\varphi = 0$ but such functions are not available. Following the procedure developed in   \cite{LMIG}, we can design GPWs adapted to \eqref{sys:gen}. These basis functions satisfy not exactly the dual equation, but will be designed to satisfy locally the approximation $L^*\varphi \approx 0$, $\varphi$ being the exponential of a polynomial. More precisely, using Taylor expansions, this subsection focuses on the design of basis functions satisfying $L^*\varphi =O(h^q)$ for any given order of approximation $q$. As in   \cite{LMIG}, the GPWs are defined by the composition of the exponential function with a polynomial. As a result, the non-linear system appearing in the design process can be solved for the Helmholtz equation. The corresponding algorithm is described here as a more general tool.

\subsubsection{Design of a basis function}
Since the design process is local, focus on a given mesh element $\Ok$ and its centroid $G\in \mathbb R^2$. 
The question to be answered is how to compute the coefficients $(\luij)$ such that the  function $\varphi= e^{P(x-x_G,y-y_G)}$ where $\displaystyle P = \sum_{0\leq i+j\leq \deg P}\luij X^iY^j$ would satisfy  $L^*\varphi =O(h^q)$. A natural answer comes from carefully canceling successively the Taylor expansion coefficients of order lower than $q$ in $L^*\varphi$. 
To simplify the notation, define the polynomial 
\enn{\begin{array}{rc}
\mathcal P_{L^*} &= -\dx^2 P - 2dr\dx\dy P - |d|^2\dy^2 P  -(\dx P)^2 - 2dr\dx P\dy P - |d|^2(\dy P)^2 \\ 
&\phantom =   -2\imath d_i(X+x_G)\dy P
\end{array}  }
   so that $L^*\varphi =\mathcal P_{L^*}(x-x_G,y-y_G)+\overline c$.
  
The system to be solved is then described as follows. The unknowns are the polynomial coefficients $(\luij)$: for a given deg of $P$, there are $N_u = \frac{(\deg P +1)(\deg P +2)}{2}$ unknowns. The equations are obtained by equating the Taylor expansion coefficients of $L^*\varphi$ of order lower than $q$ to zero: there are $N_e = \frac{q(q+1)}{2}$ equations. Because of the square and product terms, this system is not linear and thus requires further investigation.

Setting $\deg P = q+1$ provides an under-determined system and since then  $N_u-N_e = 2q+3$, $2q+3$ unknowns can be fixed conveniently to obtain an invertible system. Since it is not so obvious how to proceed for $\deg P = q$, and since the system is over-determined for $\deg P < q$, we set from now on $\deg P = q+1$. The system now reads: for $0\leq i+j\leq q-1$,
\begin{equation}\label{sys:lus}
\left\{
\begin{array}{rl}
      0&=-2\lu{2}{0} -2d_r\lu{1}{1} - 2|d|^2\lu{0}{2} -\luoz^2 -2d_r \luoz\luzo - |d|^2 \luzo^2  \\
&\phantom{=-} - 2\imath d_i x_G \luzo+\overline c (x_G,y_G) \\ 
0& = -2\lu{2}{j} - 2d_r (j+1)\lu{1}{j+1}  - 2\imath d_i x_G (j+1) \lu{0}{j+1} \qquad \qquad (j>0)\\
 \\ & \displaystyle \phantom{=-} 
 - |d|^2(j+2)(j+1)\lu{0}{j+2} 
- \sum_{l=0}^j \lu{1}{j-l}\lu{1}{l}
+\frac{1}{j!} \dy^j\overline c (x_G,y_G)
\\ &\displaystyle \phantom{=-}
-2d_r  \sum_{l=0}^j (l+1)\lu{1}{j-l}\lu{0}{l+1}-|d|^2 \sum_{k=0}^j (j-k+1)\lu{0}{j-k+1}\lu{0}{k+1},\\
    0&= -(i+2)(i+1)\lu{i+2}{j} - 2d_r (i+1)(j+1)\lu{i+1}{j+1}  \qquad \qquad (i>0)
\\ & \phantom{=-} - 2\imath d_i x_G (j+1)\lu{i}{j+1}- |d|^2(j+2)(j+1)\lu{i}{j+2} 
\\ & \phantom{=-} - 2\imath d_i(j+1)\lu{i-1}{j+1}+\frac{1}{i!}\frac{1}{j!} \dx^i\dy^j\overline c (x_G,y_G)
 \\ &\displaystyle \phantom{=-}-\sum_{k=0}^i \sum_{l=0}^j (i-k+1)(k+1)\lu{i-k+1}{j-l}\lu{k+1}{l}
\\ &\displaystyle \phantom{=-}-2d_r \sum_{k=0}^i \sum_{l=0}^j (i-k+1)(l+1)\lu{i-k+1}{j-l}\lu{k}{l+1}
\\ &\displaystyle \phantom{=-}-|d|^2\sum_{l=0}^i \sum_{k=0}^j (j-k+1)(k+1)\lu{i-l}{j-k+1}\lu{l}{k+1}.
\end{array}\right.
\end{equation}
The next step is to describe which $2q+3$ unknowns can be fixed to make this system conveniently invertible as announced. So for a given $(i,j)$, examine the corresponding equation in the system. The key point is to realize that the nonlinear terms only involve unknowns $\lu{m}{n}$ such that $m+n \leq i+j+1$, whereas the unknowns $\lu{m}{n}$ such that $m+n = i+j+2$ only appear in linear terms. As a result, if for each equation these unknowns involved in the non-linear terms are known, the resulting system to be solved would in fact be linear. That is to say the system can be written:  for $0\leq i+j\leq q-1$,
\begin{equation}\label{sys:modlu}
\left\{
\begin{array}{l}
 2\lu{2}{0}    +2d_r\lu{1}{1} + 2|d|^2\lu{0}{2} 
 \\ \phantom{=-}=  -\luoz^2 -2d_r \luoz\luzo - |d|^2 \luzo^2  - 2\imath d_i x_G \luzo
 +\overline c (x_G,y_G) \\ 
2\lu{2}{j} + 2d_r (j+1)\lu{1}{j+1}   + |d|^2(j+2)(j+1)\lu{0}{j+2}
 \\  \displaystyle \phantom{=-} 
 = 
 -2\imath d_i x_G (j+1) \lu{0}{j+1}
- \sum_{l=0}^j \lu{1}{j-l}\lu{1}{l}
-2d_r  \sum_{l=0}^j (l+1)\lu{1}{j-l}\lu{0}{l+1}
\\ \displaystyle \phantom{===-}-|d|^2 \sum_{k=0}^j (j-k+1)\lu{0}{j-k+1}\lu{0}{k+1}+\frac{1}{j!} \dy^j\overline c (x_G,y_G), \ (j>0)
\\
    (i+2)(i+1)\lu{i+2}{j} + 2d_r (i+1)(j+1)\lu{i+1}{j+1}+ |d|^2(j+2)(j+1)\lu{i}{j+2}
\\  \phantom{=-}= - 2\imath d_i x_G (j+1)\lu{i}{j+1} -2\imath d_i(j+1)\lu{i-1}{j+1}
+\frac{1}{i!}\frac{1}{j!} \dx^i\dy^j\overline c (x_G,y_G)
 \\ \displaystyle \phantom{===-}-\sum_{k=0}^i \sum_{l=0}^j (i-k+1)(k+1)\lu{i-k+1}{j-l}\lu{k+1}{l}
\\ \displaystyle \phantom{===-}-2d_r \sum_{k=0}^i \sum_{l=0}^j (i-k+1)(l+1)\lu{i-k+1}{j-l}\lu{k}{l+1}
\\ \displaystyle \phantom{===-}-|d|^2\sum_{l=0}^i \sum_{k=0}^j (j-k+1)(k+1)\lu{i-l}{j-k+1}\lu{l}{k+1},
 \quad\quad (i>0).
\end{array}\right.
\end{equation}

So first of all, suppose that the right hand side of \eqref{sys:modlu} is known. Now for a given level $\mathcal L$, $0\leq \mathcal L\leq q-1$, there is a subsystem of \eqref{sys:modlu} made of the $\mathcal L+1$ equations corresponding to $(i,j) = (i,\mathcal L-i)$. The unknowns of this sub-system are the $\lu{i}{\mathcal L+2-i}$ with $0\leq i\leq\mathcal L+2 $. So this sub-system has
\begin{itemize}
\item $\mathcal L +1$ equations,
\item $\mathcal L + 3$ unknowns.
\end{itemize}
Moreover the subsystem is tridiagonal since in each equation indexed by $i$ the unknowns involved are $\lu{i+2}{\mathcal L -i}$, $\lu{i+1}{\mathcal L -i +1}$ and $\lu{i}{\mathcal L -i +2}$. As a consequence, one easy way to get an invertible subsystem is to fix the two first unknowns $\lu{0}{\mathcal L+2}$ and $\lu{1}{\mathcal L +1}$. The solution of the subsystem then simply reads : for all $i \text{ from } 0 \text{ to } \mathcal L$,	
\begin{equation}\label{eq:luRHS}
\begin{array}{l}
\lu{i+2}{\mathcal L -i} = \frac 1{(i+2)(i+1)} 
\left(
RHS(i) - 2d_r (i+1)(\mathcal L-i+1)\lu{i+1}{\mathcal L-i+1}
\right.
 \\
\phantom{\lu{i+2}{\mathcal L -i} = \frac 1{(i+2)(i+1)} 
(
RHS(i).}
\left.
-  |d|^2(\mathcal L-i+2)(\mathcal L-i+1)\lu{i}{\mathcal L-i+2} 
\right),
\end{array}
\end{equation}
$RHS(i)$ being the right hand side of the corresponding Equation $(i,\mathcal L-i)$ in System \eqref{sys:modlu}.

 Then suppose the right hand side of a subsystem for a given level $\mathcal L$ is known. Since this subsystem can be solved thanks to \eqref{eq:luRHS}, it is clear that the right hand side of the subsystem made of the $\mathcal L+2$ equations corresponding to $(i,j) = (i,\mathcal L+1-i)$ is known. The only remaining step is to initialize this induction process for the level $\mathcal L = 0$. At this level the subsystem is only one equation, namely
 \begin{equation}
  2\lu{2}{0}    +2d_r\lu{1}{1} + 2|d|^2\lu{0}{2} =  -\luoz^2 -2d_r \luoz\luzo - |d|^2 \luzo^2  - 2\imath d_i x_G \luzo  +\overline c (x_G,y_G) .
 \end{equation}
So fixing the unknowns $\luoz$ and $\luzo$ fixes the right hand side.

To summarize, for each level $\mathcal L$ two unknowns have to be fixed, so $2q$ unknowns, plus two unknowns for the level $\mathcal L = 0$. Notice moreover that the unknown $\lu{0}{0}$ does not appear in the system, so it can be fixed without any consequence on the previous reasoning. Altogether these are $2q+3$ unknowns fixed, corresponding to the difference $N_u-N_e$ between the number of unknowns and the number of equations of the initial system.

Finally, solving system \eqref{sys:lus} can be described by the following algorithm:
\begin{itemize}
\item fix $\lu{0}{0}$
\item for level  $\mathcal L = 0$
\begin{itemize}
\item fix  $\luoz$ and $\luzo$
\item  compute
\end{itemize}
$$  
\begin{array}{l}
\lu{2}{0}  =\frac12\left(-2d_r\lu{1}{1} - 2|d|^2\lu{0}{2}   -\luoz^2 -2d_r \luoz\luzo - |d|^2 \luzo^2 \right. 
\\ \phantom{\lu{2}{0}  =\frac12(-2d_r\lu{1}{1} - 2|d|^2\lu{0}{2}   -\luoz^2}
\left.- 2\imath d_i x_G \luzo  +\overline c (x_G,y_G)\right)
\end{array}$$
\item for all levels $\mathcal L$ from $1$ to $q-1$
\begin{itemize}
\item fix $\lu{0}{\mathcal L+2}$ and $\lu{1}{\mathcal L +1}$
\item compute 
\end{itemize}
\end{itemize}
\begin{equation*}
\begin{array}{l}
 \displaystyle \lu{2}{\mathcal L}
   =\frac{1}{2 }\Bigg( - |d|^2(\mathcal L+2)(\mathcal L+1)\lu{0}{\mathcal L+2}   -2\imath d_i x_G (\mathcal L+1) \lu{0}{\mathcal L+1}
\\  \displaystyle \phantom{  \lu{2}{\mathcal L} =\frac{1}{2 } 0 }
 -  2d_r (\mathcal L+1)\lu{1}{\mathcal L+1} 
 - \sum_{l=0}^{\mathcal L} \lu{1}{\mathcal L-l}\lu{1}{l}
-2d_r  \sum_{l=0}^\mathcal L (l+1)\lu{1}{\mathcal L-l}\lu{0}{l+1}
\\  \displaystyle \phantom{  \lu{2}{\mathcal L}=\frac{1}{2 }0}
   -|d|^2 \sum_{k=0}^\mathcal L(\mathcal L-k+1)\lu{0}{\mathcal L-k+1}\lu{0}{k+1}+\frac{1}{\mathcal L!} \dy^{\mathcal L}\overline c (x_G,y_G)\Bigg),
\end{array}
\end{equation*}
\begin{equation*}
\begin{array}{rr}
\lu{i+2}{\mathcal L-i}& \displaystyle= \frac{1}{(i+2)(i+1) }\Bigg( - |d|^2(\mathcal L-i+2)(\mathcal L-i+1)\lu{i}{\mathcal L-i+2}    \phantom{rrrrrrrrr}
 \\(i>0)& \displaystyle
 - 2d_r (i+1)(\mathcal L-i+1)\lu{i+1}{\mathcal L-i+1} 
 - 2\imath d_i x_G (\mathcal L-i+1)\lu{i}{\mathcal L-i+1} 
  \\& \displaystyle 
 -2\imath d_i(\mathcal L-i+1)\lu{i-1}{\mathcal L-i+1}
+\frac{1}{i!}\frac{1}{(\mathcal L-i)!} \dx^i\dy^{\mathcal L-i}\overline c (x_G,y_G)
 \\& \displaystyle -\sum_{k=0}^i \sum_{l=0}^{\mathcal L-i} (i-k+1)(k+1)\lu{i-k+1}{\mathcal L-i-l}\lu{k+1}{l}
\\& \displaystyle -2d_r \sum_{k=0}^i \sum_{l=0}^{\mathcal L-i} (i-k+1)(l+1)\lu{i-k+1}{\mathcal L-i-l}\lu{k}{l+1}
\\& \displaystyle -|d|^2\sum_{l=0}^i \sum_{k=0}^{\mathcal L-i} (\mathcal L-i-k+1)(k+1)\lu{i-l}{\mathcal L-i-k+1}\lu{l}{k+1}
\Bigg).
\end{array}
\end{equation*}

\subsubsection{Normalization of a set of basis functions}
The choice of polynomial coefficients $\left\{ \lu{i}{j}, i\in \{0,1\}, 0\leq j \leq q+1-i \right\}$ is then the only remaining step to completely design a GPW. This provides a tool to define locally not only one but a set of basis functions.

First of all, in order to simplify the computational steps, it is natural to set as many coefficients as possible equal to zero. Keeping in mind the example of a classical plane wave, the idea proposed here is to fix $(\luoz,\luzo)$ - depending on an angle $\theta$ -  to ensure that $\lu{2}{0} =0$, while all the other coefficients of $\left\{ \lu{i}{j}, i\in \{0,1\}, 0\leq j \leq q+1-i \right\}$ are set equal to zero. As a consequence the resulting GPW $\varphi$ satisfies $\varphi = \exp (\luoz \tilde x +\luzo \tilde y + O\left((\tilde x ,\tilde y)^3\right)$. From the expression of $\lu{2}{0}$ provided in the algorithm, one can see that this is easily done by fixing $\luzo$ and setting
\begin{equation}\label{eq:luoz}
 \luoz = -d_r \luzo +\sqrt{-(d_i^2 \luzo^2+2\imath d_i x_G\luzo) + \overline c (x_G,y_G)}.
 \end{equation}
The design of a local set of basis functions is then achieved, following the plane wave example, by fixing the coefficient $\luzo=\frac{\sqrt{ \overline c (x_G,y_G)}}{d_i} \sin\theta$, for equi-spaced values of $\theta \in [0, 2\pi)$.

\subsubsection{Summary }

On a given element of the mesh $\Omega_k$, define the  basis functions $\{\varphi_k^l\}_{1\leq l \leq p(k)}$ by:
\begin{itemize}
\item[$\bullet$] $\overrightarrow g_k=(x_G,y_G)$ is the center of gravity , 
\item[$\bullet$] $p(k)$ is the number of basis functions,
and 
$\forall l$ such that $1\leq l \leq p(k)$,
\begin{itemize}
\item  define $\theta_l = \frac{2\pi}{p(k)}(l-1)$,
\item define $\luzo^l = \frac{\sqrt{ \overline c (x_G,y_G)}}{d_i} \sin\theta_l$,
\item compute $\luoz^l$ from \eqref{eq:luoz},
\item set the other fixed coefficients 

$\left\{ \lu{i}{j}^l, i\in \{0,1\},  0\leq j \leq q+1-i , i+j\neq 1\right\}$ to zero,
\end{itemize}
\item[$\bullet$] $q(k)$ is the approximation parameter on $\Omega_k$,
 $\forall l$ such that $1\leq l \leq p(k)$ and $\forall (i,j)$ such that $0\leq i+j\leq q+1$
\begin{itemize}
\item compute the remaining polynomial coefficients according to the algorithm
\item form the corresponding GPW $\varphi_k^l =\exp \sum \luij^l \tilde x ^i \tilde y ^j$.
\end{itemize} 
\end{itemize}
To obtain functions defined on the whole domain $\Omega$, these basis functions $\{\varphi_k^l\}_{1\leq l \leq p(k)}$ are set to be zero on $\Omega\backslash \Omega_k$. This process defines a set of basis functions on $\Omega$, namely
\enn{
\mathcal E = \cup_k\mathcal E_{\overrightarrow g_k} (N,p(k),q(k)) \textrm{ where } \left.\mathcal E_{\overrightarrow g_k} (N,p(k),q(k))\right|_\Ok=\left\{ \varphi_k^l \right\}_{1\leq l \leq p(k)},
}
that will be used to discretize the UWVF \eqref{UWVF}.

Thanks to the analysis presented in this section, computing the polynomial coefficients of a GPW does not require to solve any system: it reduces to applying the induction formula. All the polynomial coefficients necessary to define the function space $\mathcal E$ can be precomputed and stored, so that the evaluation of any GPW simply requires to read the corresponding coefficients from the precomputed table. The implementation that is used to produce the results displayed in the next section provides the following timing for $p=7$ and $q=4$: for a mesh of $34036$ elements, the computation of the polynomial coefficients table takes 235 seconds while it takes 3 seconds for the classical plane waves equivalent computation ; for a mesh of $78574$ elements, the computation of the polynomial coefficients table takes 701 seconds while it takes 7 seconds for the classical plane waves equivalent computation. These two sizes of mesh are typically the smaller and bigger sizes used in the next section. These timings are negligible compared to the time required to build the matrix.

\section{Numerical simulation}
\label{sec:NS}
The aim of this section is to show numerical evidence of waves propagating through the mode conversion point. In our model, the propagative zones are defined by $\{x<0 \text{ and } y>-x \}$ and $\{x>0 \text{ and } y<-x \}$. A series of test cases will be designed to observe some features of this model, with the concern of observing the incoming wave and avoiding non-physical reflections.

\subsection{Definition of the test cases}
Since the mode conversion point in the $(x,y)$ coordinate system is the origin, the computational domain is centered at the origin. Remember that the medium is propagative if $x(x+y)<0$ and evanescent if $x(x+y)>0$. The propagative and evanescent zones are separated by two cut-offs, $x=0$ and $x+y=0$.

As already mentioned, the variable coefficients are key to the mode conversion phenomenon. The two only variable coefficients are
\begin{itemize}
\item the zeroth order term, varying as $x(x+y)$,
\item the first order term, varying as $x$.
\end{itemize}
 To observe an incoming wave before and after crossing the mode conversion region, these variable coefficients are set to constant values away from the origin. That is to say that the zeroth and first order coefficients will be allowed to vary only in a vicinity of the mode conversion point, while they will be set to a constant value in a zone further away from this point. 
\begin{figure}
\begin{center}
\includegraphics[clip,trim=3cm 1cm 3cm 1cm,height=4.5cm]{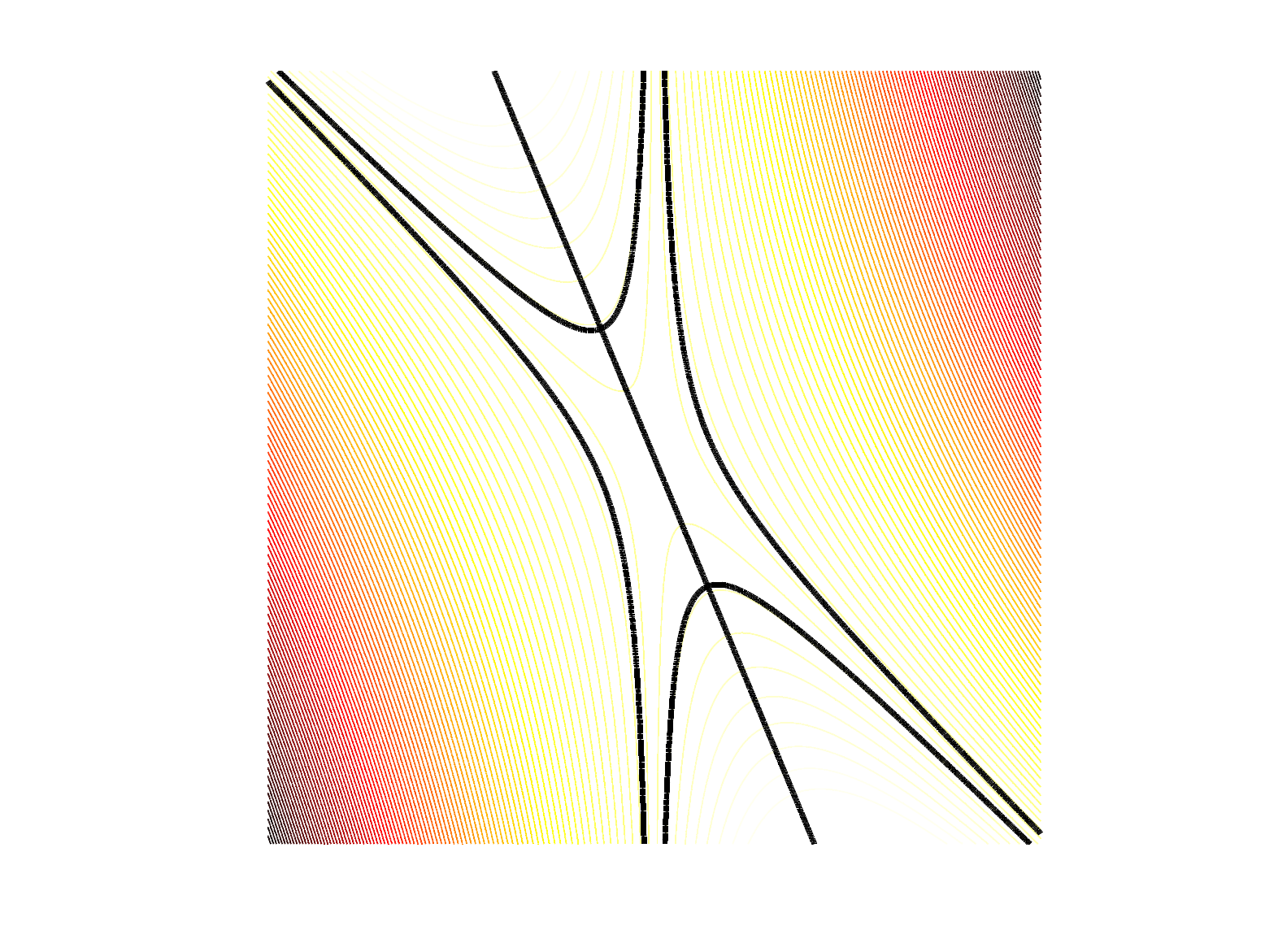}
\includegraphics[clip,trim=3cm 1cm 3cm 1cm,height=4.5cm]{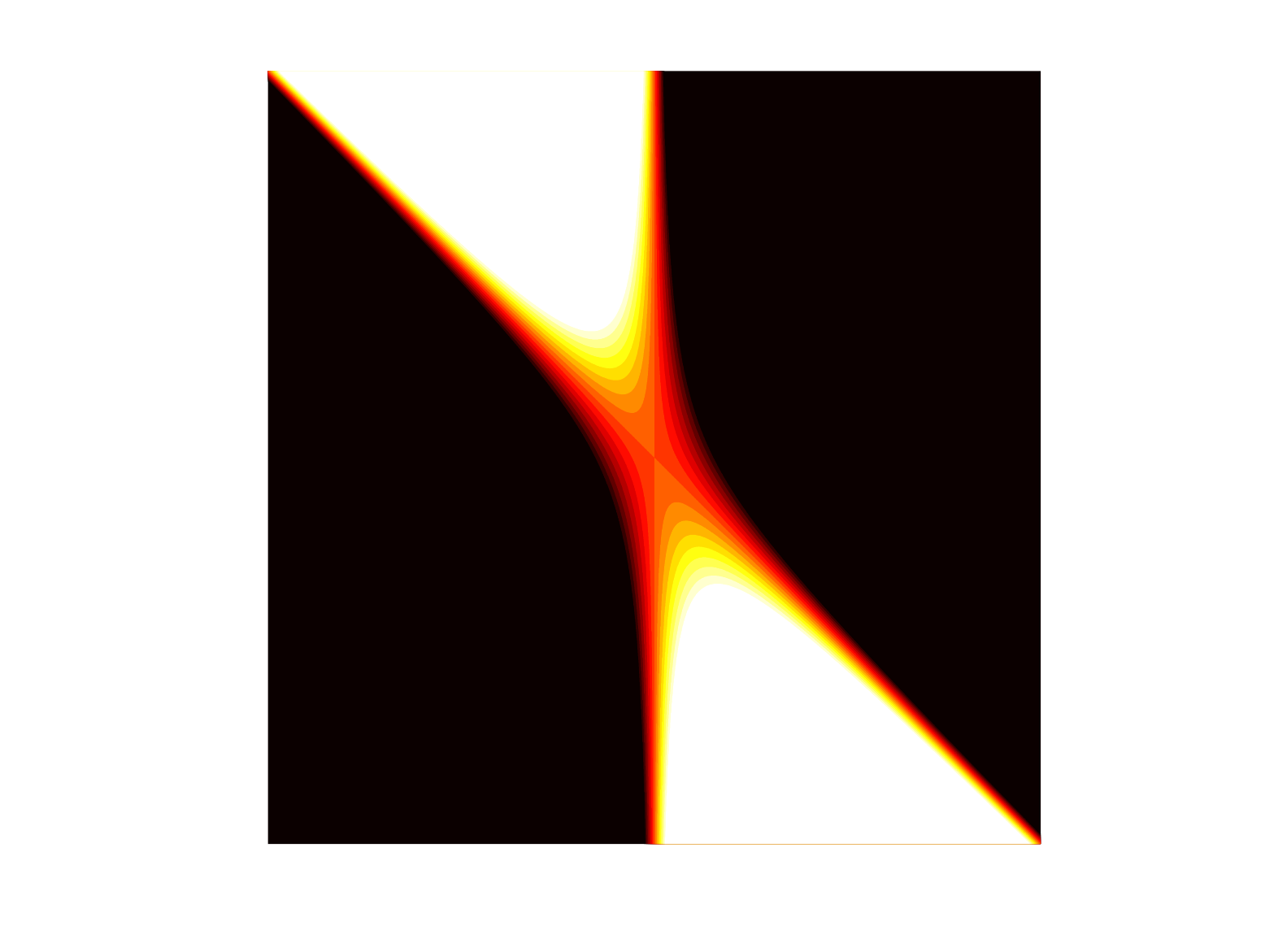}
\includegraphics[clip,trim=3cm 1cm 3cm 1cm,height=4.5cm]{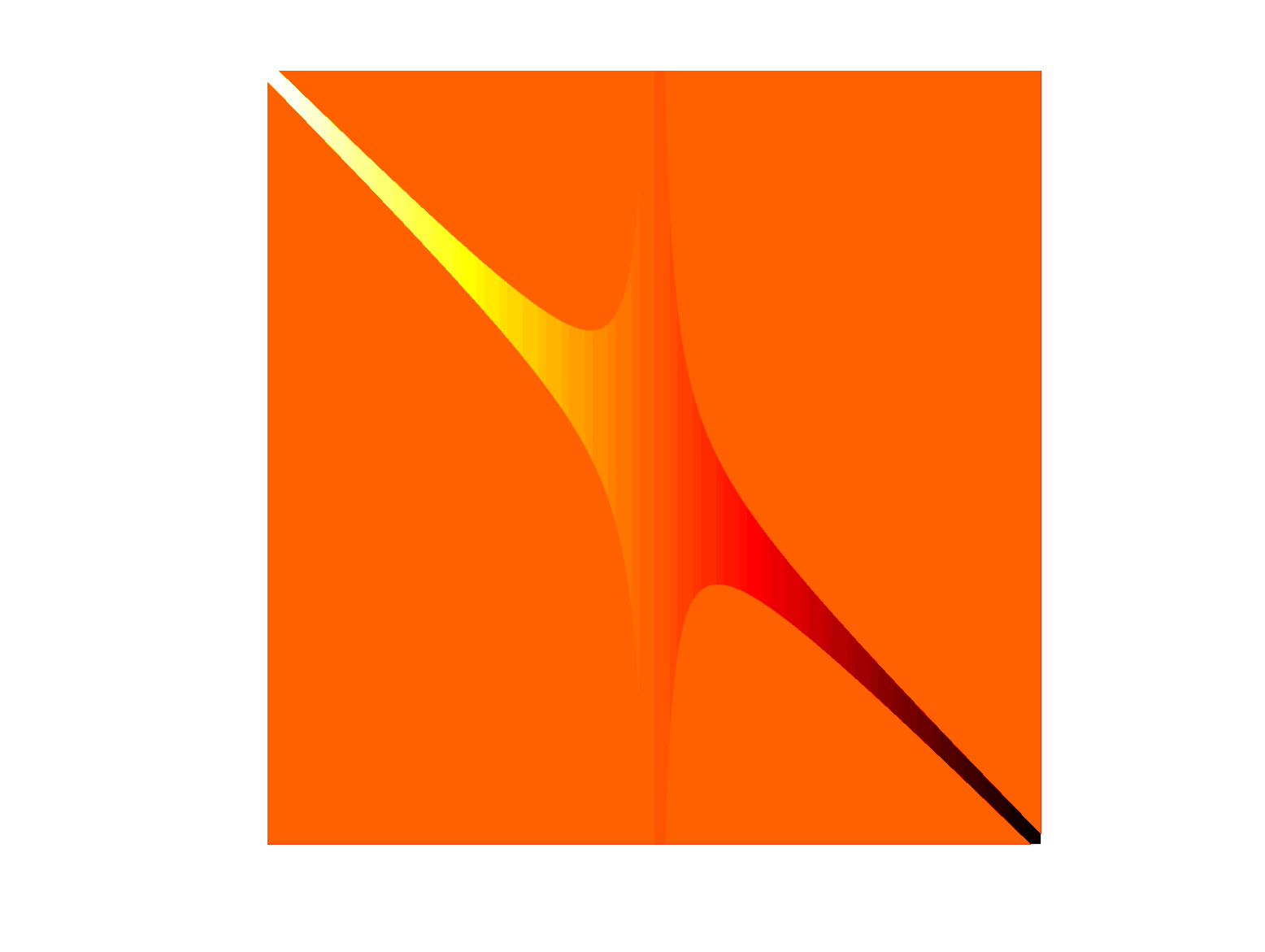}
\end{center}
\caption{Functions $f$, $\tilde f$ and $\tilde g$ defined in \eqref{eq:tftg}, for $y_K = 6$. Left: Level curves of $f$, highlighting the $C = f(x_K,y_K)$ and $\pm C$ level curves together with the line of steepest gradient in the propagative zone. Middle: In the black zone $\tilde f$ is constant and positive, it is included in the evanescent zone. In the white zone $\tilde f$ is constant and negative, it is included in the propagative zone. Right: $\tilde g$ is non-zero only in the region where $|x(x+y)|\leq -C$.}
\label{fig:VC}
\end{figure}  In order to introduce only one artificial interface between the constant coefficient zone and the varying coefficient zone, the constant zone is chosen to be the same for both coefficients. It is fixed at level curves of the function $f(x,y) = x(x+y)$, for a given transition point $(x_K,y_K)$, see Figure \ref{fig:VC}. 
 The size of the varying coefficient zone clearly varies with the distance between the transition point and the origin. 
  In the propagative zone, that is to say $\{ (x,y) \ s.t. \ f(x,y)<0 \}$, the steepest gradient of $f$ is along the line $y=-(\sqrt 2 + 1)x$. For a given point $(x_K,y_K)$ along this line and $C = f(x_K,y_K)$, the constant zone has four connected components defined by $\{(x,y) \text{ s.t.}  f(x,y)<C \text{ and } x<0  \}$, $\{(x,y) \text{ s.t.}  f(x,y)<C \text{ and } x>0  \}$, $\{(x,y) \text{ s.t. } f(x,y)>-C \text{ and } x<0  \}$ and $\{(x,y) \text{ s.t. } f(x,y)>-C \text{ and } x>0  \}$.  In the model studied in this work, the two former components are in the propagative zone while the two latter ones are in the evanescent zone. The only remaining connected component of the plane is the varying zone. It can also be seen as a transition zone around the mode conversion point. In the constant zone the first order term's coefficient is set to zero, and zeroth order term $x(x+y)$ is replaced by either $\pm C$ in order to be continuous. That is to say that the coefficients used for numerical purposes are
 \begin{equation}
 \label{eq:tftg}
\tilde f (x,y) = \left\{\begin{array}{l}
 -C   \text{ if } x(x+y)\geq -C, \\
 \phantom - C   \text{ if } x(x+y)\leq \phantom -C, \\
 x(x+y)  \text{ otherwise},
 \end{array}\right.
 \tilde g (x,y) =
 \left\{\begin{array}{rl}
 x& \text{ if }|x(x+y)|\leq -C, \\
 0 & \text{otherwise.}
 \end{array}\right.
 \end{equation}
This function $\tilde f$ is continuous while $\tilde g$ is discontinuous along the level curves $f(x,y) = \pm C$. Figure \ref{fig:VC} 
 shows the modified variable coefficients. As a consequence of this setting, a wave propagating from the constant propagative zone toward the mode conversion point is expected to propagate until it reaches the varying zone. The wave can then split into a reflected wave and a transmitted wave propagating at least partially through the varying zone. Mode conversion then corresponds to the existence of a transmitted wave exiting this transition zone on the other connected component of the constant propagative zone.
  
The antenna is made of a wave guide of width $l_0$ and length $4l_0$, plus a horn, as described in Figure \ref{fig:CD}. 
In order to send a wave toward the mode conversion point, this antenna is placed facing the origin, in the $\{x<0\}$ and $\{y>-x \}$ region.
A given transition point $(x_K,y_K)$ is fixed, defining the amplitude of the transition zone. Then the vertical position of the antenna is determined so that if it is along the steepest gradient axis, the distance between the end of the horn and the transition point is $12l_0$. 
An example of computational domain with the previously described features is also displayed in Figure \ref{fig:CD}.
\begin{figure}
\begin{center}
\includegraphics[clip,trim=0cm 0cm 0cm 0cm,height=4cm]{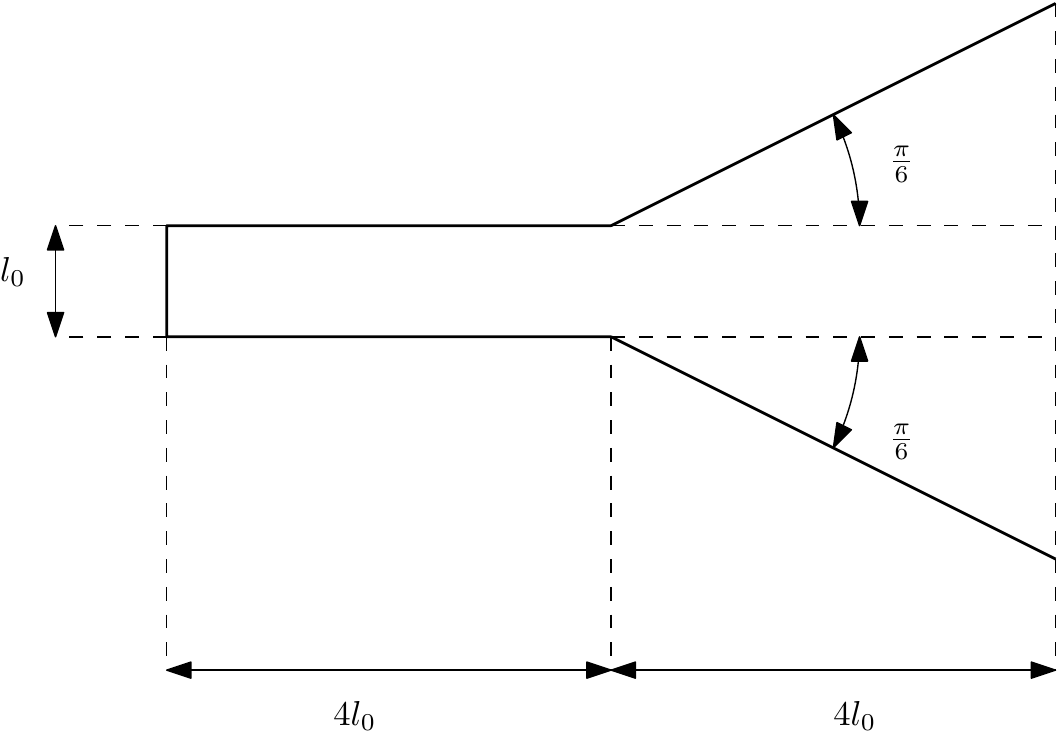}
\includegraphics[clip,trim=1.8cm .5cm 3.5cm 1cm,height=7cm]{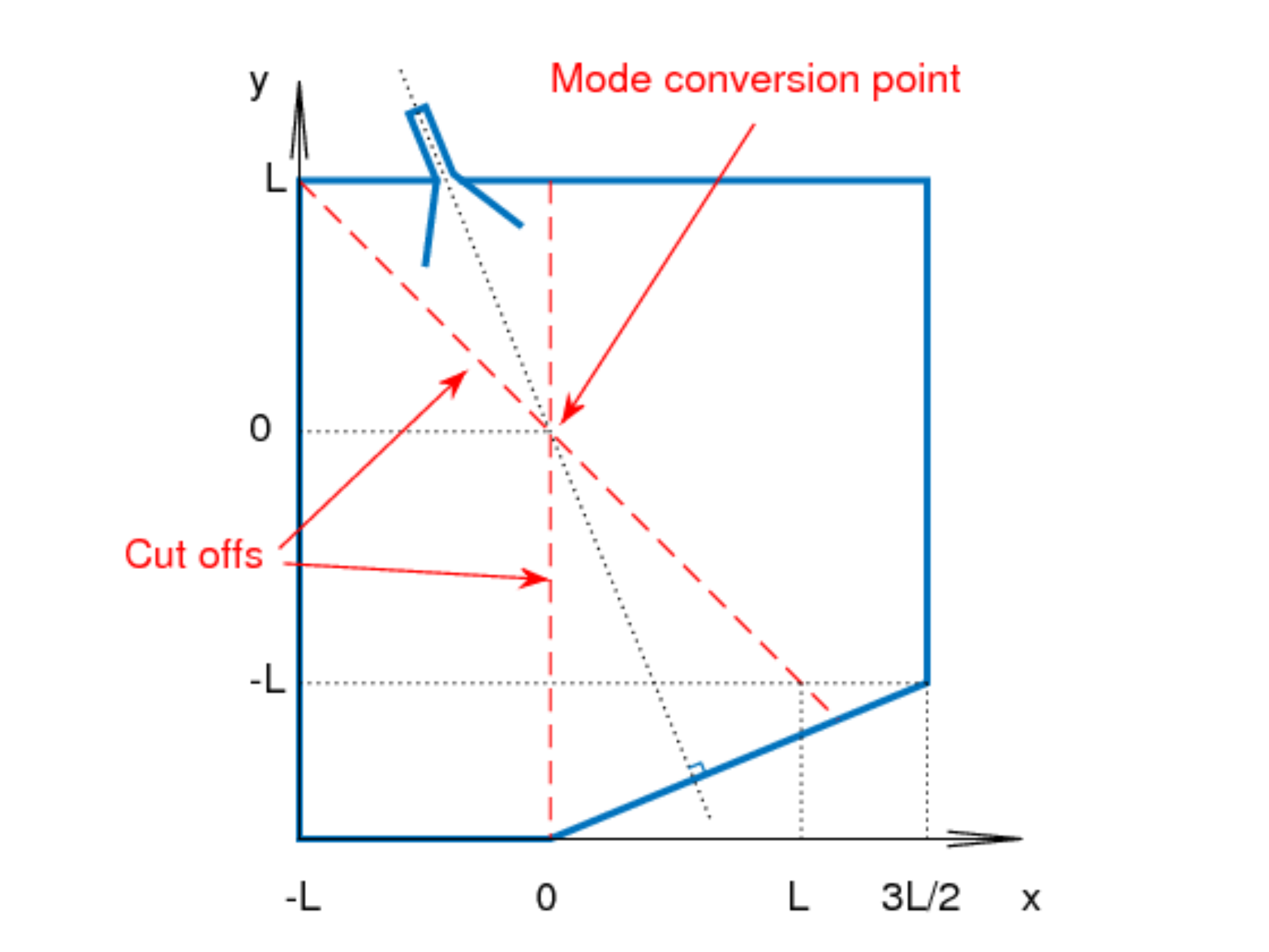}
\end{center}
\caption{Left: Geometry of the antenna that will be placed in the propagative zone facing the mode conversion point, with different incident angles. Right: Computational domain. The blue heavy line represents the boundary of the domain. The red dashed lines represent the two cut-off lines, and cross at the mode conversion point. The antenna is placed on the top boundary, in the propagative zone.}
\label{fig:CD}
\end{figure}  

Then the antenna can be placed along any axis with an angle $\theta$ with respect to the $y$ axis, $\theta\in[0,\pi/4]$, as shown in Figure \ref{fig:A}.
\begin{figure}
\begin{center}
\includegraphics[clip,trim=2cm 1cm 1.3cm 1cm,height=3.3cm]{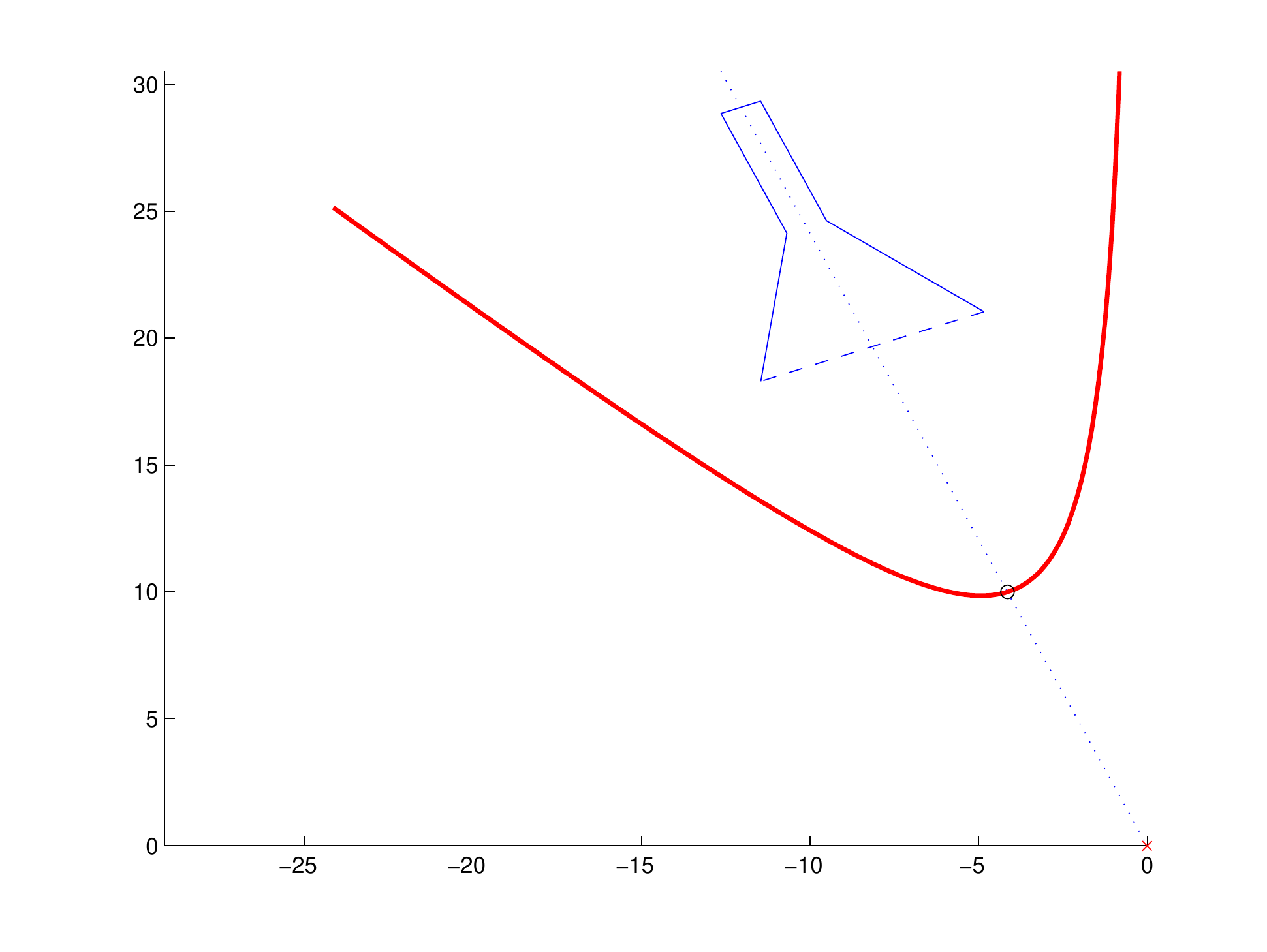}
\includegraphics[clip,trim=2cm 1cm 1.3cm 1cm,height=3.3cm]{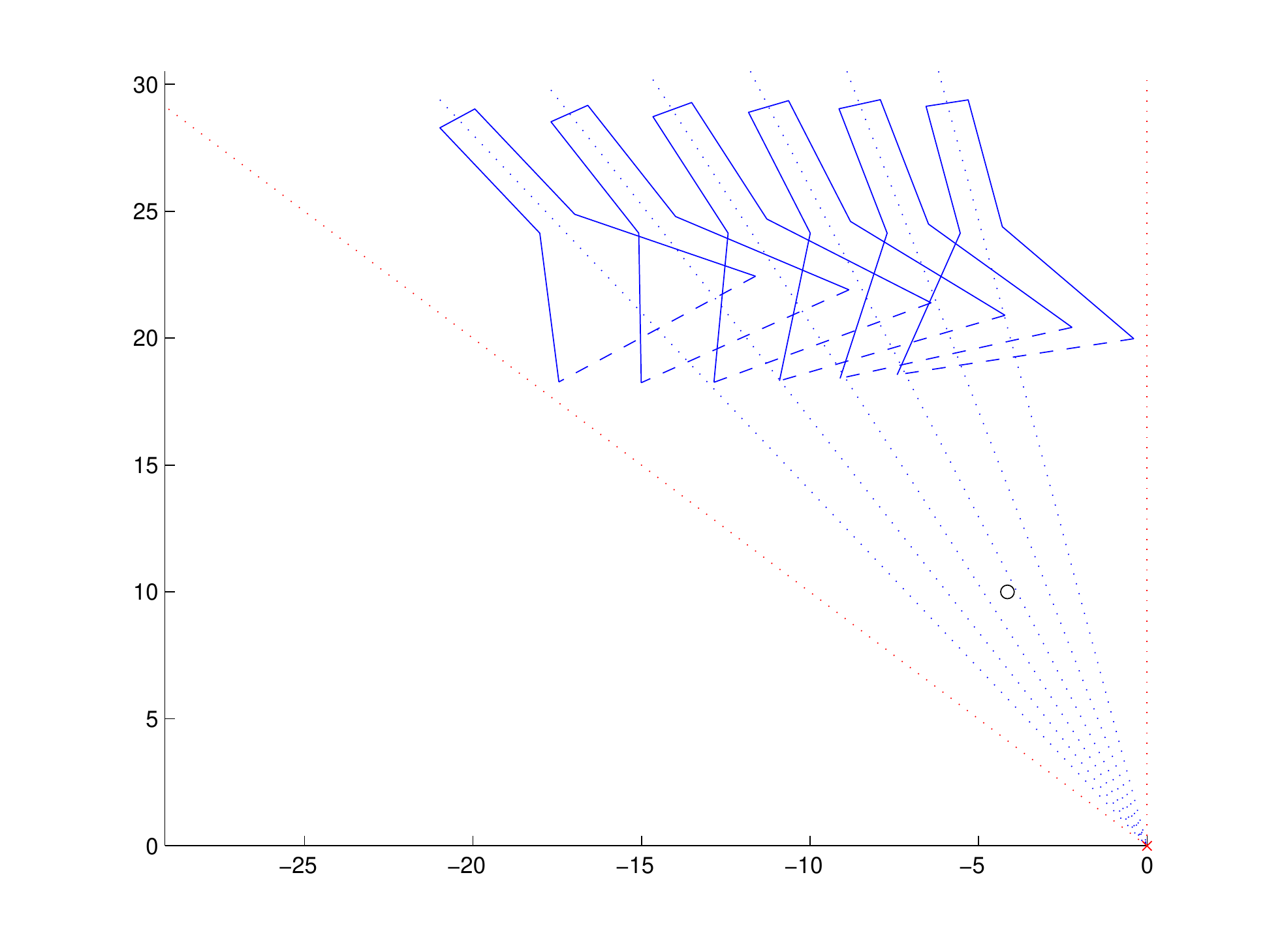}
\includegraphics[clip,trim=2cm 1cm 1.3cm 1cm,height=3.3cm]{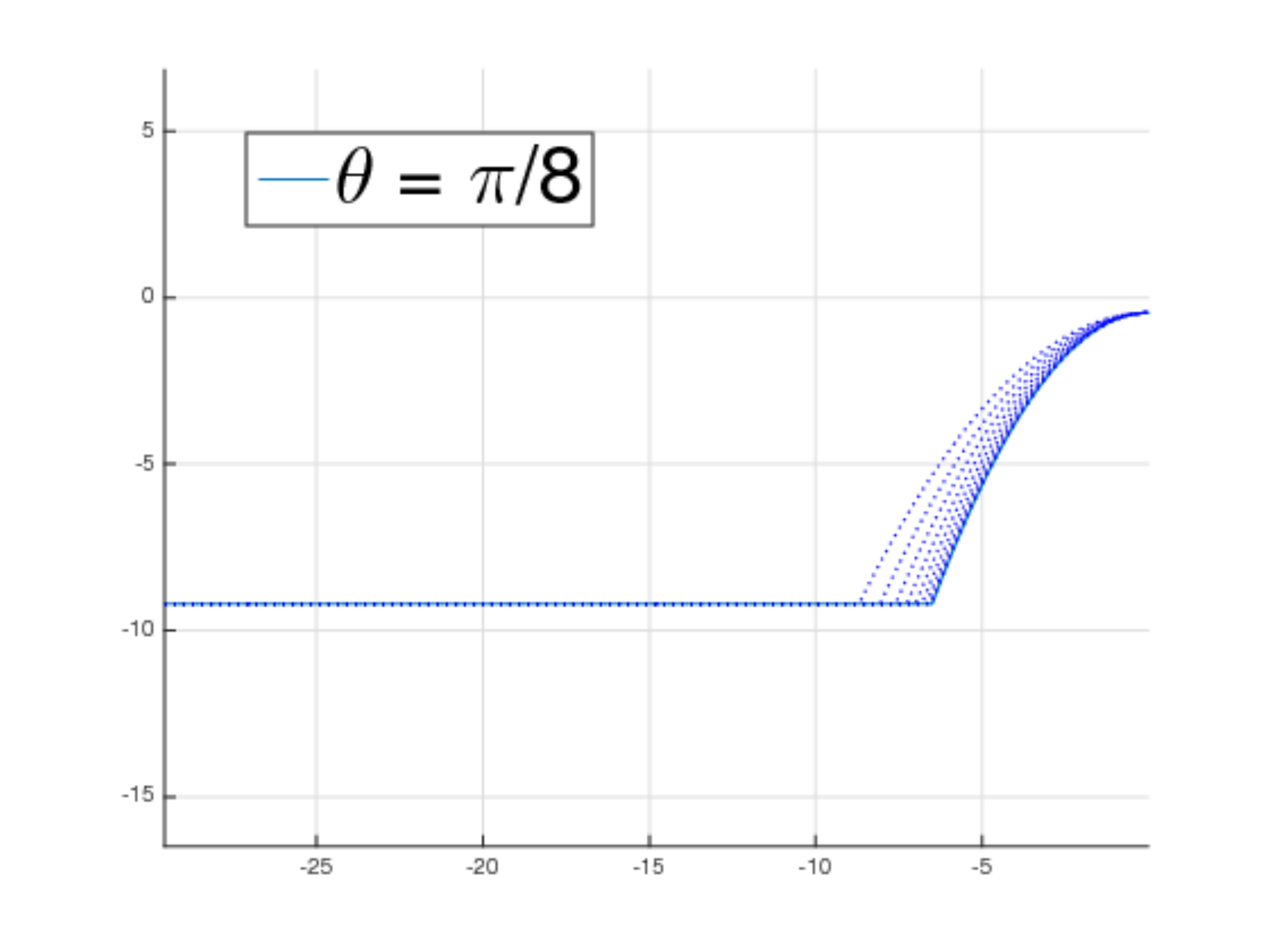}
\end{center}
\caption{
Left and middle: Antennas facing the mode conversion point, in the top left quarter of the computational domain. The mode conversion point is the bottom right corner of both graphs.
The point $(x_K,y_K)$ is the black circle while the mode conversion point is red cross. Left: Antenna along the steepest gradient axis. The transition line defined by $x(x+y) = x_K(x_K+y_K)$ and $x<0$ is the thick red line. Right: Different positions of the antenna at different angles in the propagative zone.
Right: Real part of the coefficient $c(x(t),y(t)) = 1+1/\mu+\sin\theta(\sin\theta-\cos\theta)t^2$ along the main axis of the antenna, for different incident angles, computed with $d=2-\imath$ and $y_K=6$.}
\label{fig:A}
\end{figure}  
Along the main axis of the antenna, between the antenna and the mode conversion region, the zeroth order coefficient $c(x(t),y(t)) = 1+1/\mu+\tilde f(x(t),y(t))$ is constant on the antenna side and then varying as a quadratic function until the mode conversion point. It is odd as a function of t, the parameter describing the position along the steepest gradient line.  Figure \ref{fig:A} also 
 represents this coefficient $c$ along the main axis of the antenna, for different incident angles $\theta$. The antenna is described  by $x(t)= (\sin\theta )t$ and $y(t)= -(\cos\theta) t$, the parameter $t$ being negative between the antenna and the mode conversion point, and $0$ at the the mode conversion point. 

The rest of the computational domain is a rectangle aligned with the $x$ and $y$ axis, except for the bottom right corner which is cut perpendicularly to the steepest gradient line $y=-(\sqrt 2+1)x$. This cut is meant to avoid artificial reflections of waves outgoing from the transition zone toward the $\{x>0 \text{ and } y<-x\}$ region. 

The parameter $\mathcal Q$ in the boundary condition \eqref{sys:gen} determines the type of boundary condition. On the bottom wall of the antenna, an incoming plane wave, with wave length $l_0$, is sent into the domain by a Robin type boundary condition, setting $\mathcal Q=0$. On the other walls of the antenna, a perfect conductor type boundary condition is set by $\mathcal Q=-1$ and $g=0$. Absorbing boundary conditions are set on on all the other boundaries of the domain, setting $\mathcal Q=0$ and $g=0$.

\subsection{Results and comments}

A series of results are proposed to illustrate the mode conversion phenomenon in the 2D model introduced in this work. Remember that by definition of the O- or X-mode cut-offs, a pure wave of the corresponding type can only propagate on one side of the cut-off. As a measure of mode conversion, define the transmission coefficient  computed here as the relative difference between the maximum wave amplitude before the transition zone, i.e. on $\{x< 0\} \cap \{x(x+y)< C\} $, and after the transition zone,  i.e. on $\{x> 0\} \cap \{x(x+y)< C\} $.
All the following results were computed with $p=7$ basis functions per element, $q=4$ as the approximation order and with either $d=-2-\imath$ or $d = -\tan(\pi/8)-\imath$. The absolute value of the solution will always be represented with $200$ level curves. The transition lines between the constant and varying coefficient zones are indicated in blue. Note that the computational domain is meshed automatically by only imposing the mesh size, and so the mesh does not resolve the boundaries of the propagative and absorbing region, or the zones of constant and variable coefficients.

To present the type of results obtained, Figure \ref{fig:thetamain} displays two solutions computed with the incident angle $\theta = \pi/8$ and $d = -2-\imath$, for both $y_K=7$ and $y_K = 12$. For the $y_K=7$ case, the antenna parameter is $l_0 \approx 1.82$, the mesh is generated for a mesh size of $0.64$, and is made of $36548$ triangles. For the $y_K=12$ case, the antenna parameter is $l_0 \approx 1.06$, the mesh is generated for a mesh size of $0.39$, and is made of $63456$ triangles. One can observe the same type of behavior: the outgoing wave propagates from the antenna, driven by the two cut-offs toward the mode conversion point. Reflection can be observed clearly in the transition zone where the coefficients are varying. Comparing these two solutions naturally suggests that the size of the transition zone plays an important role in the transmission process.
\begin{figure}
\begin{center}
\includegraphics[clip,trim=2cm 1cm 1.5cm 1cm,height=4cm]{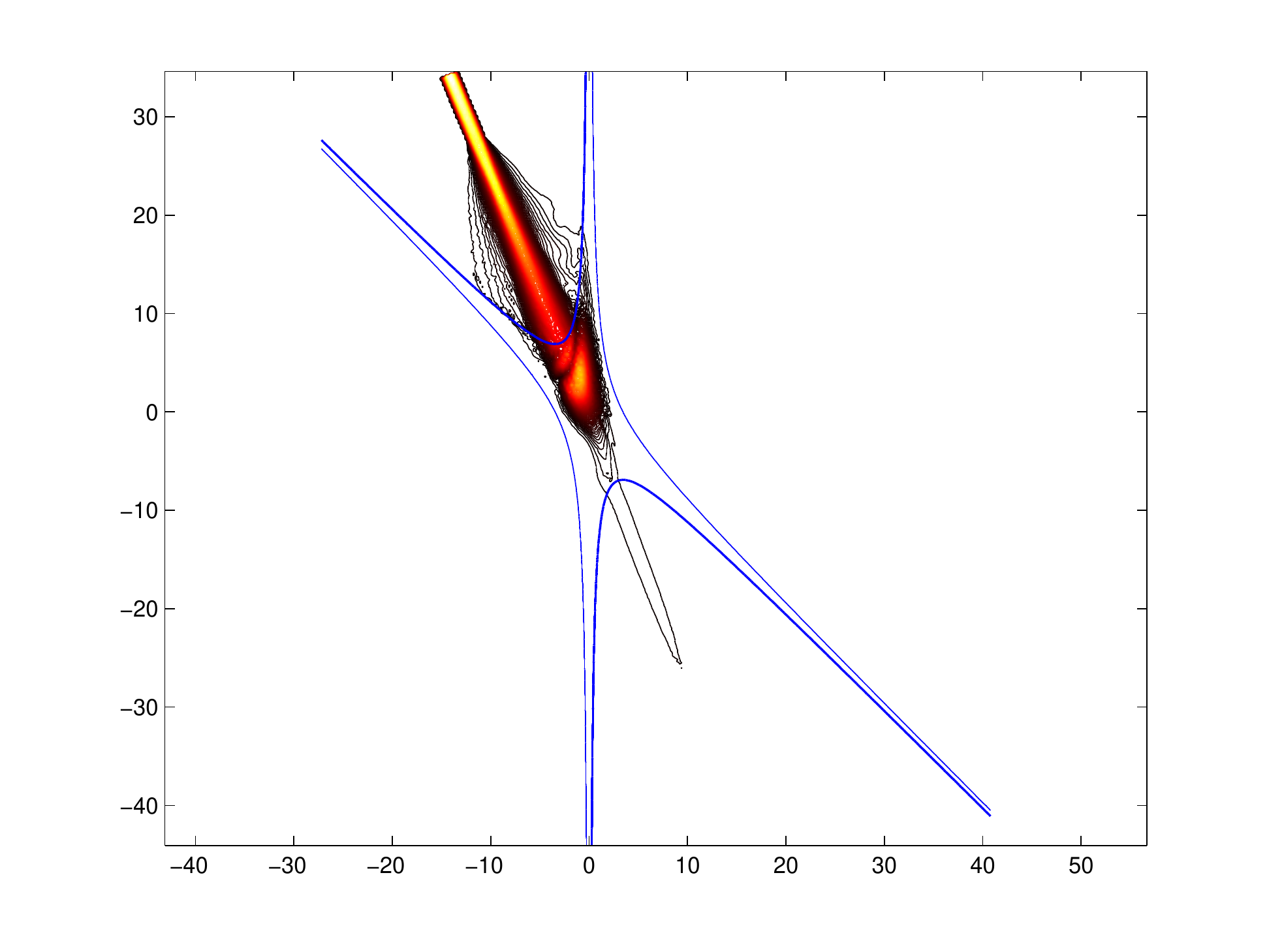}
\includegraphics[clip,trim=2cm 1cm 1.5cm 1cm,height=4cm]{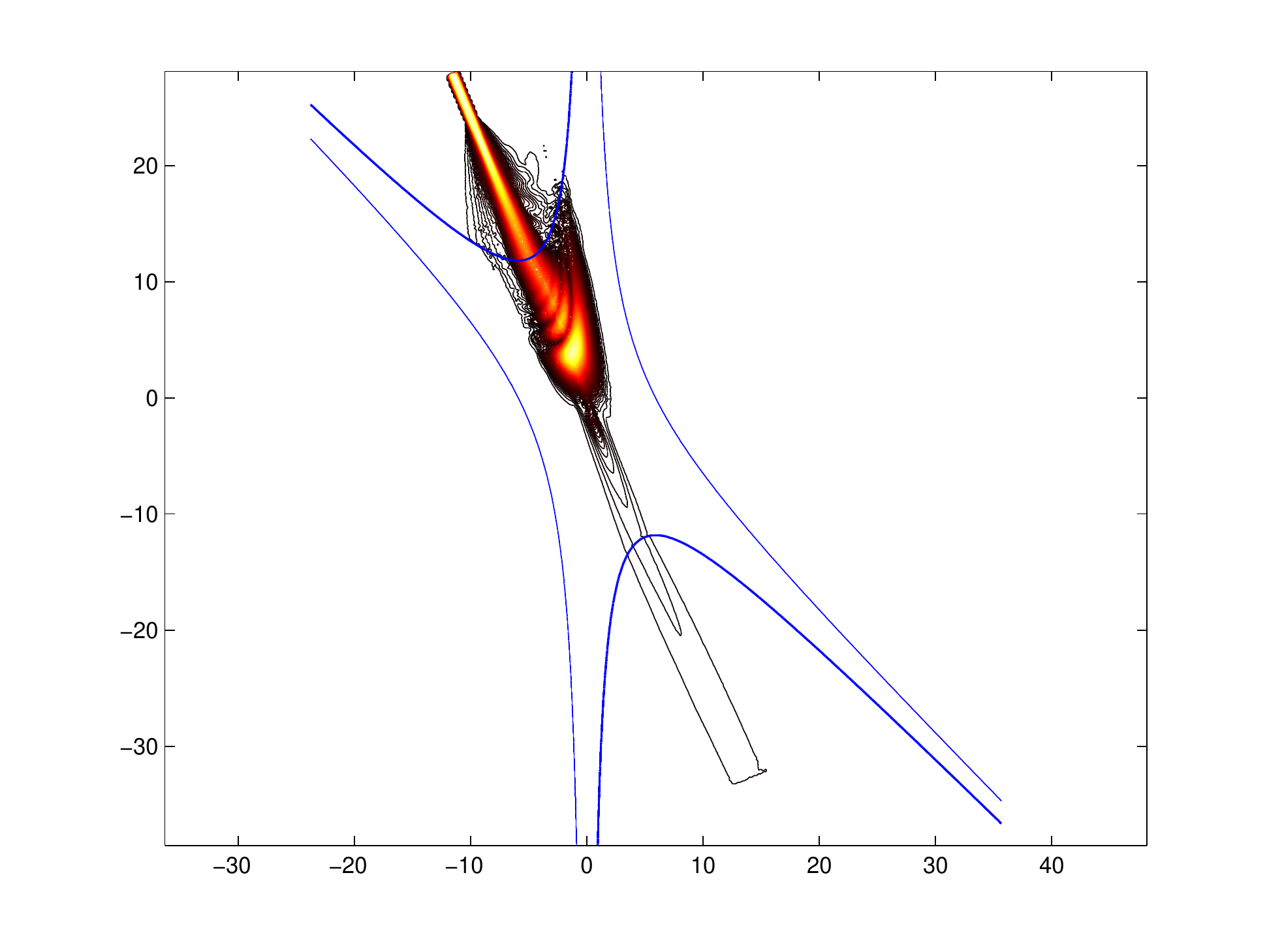}
\end{center}
\caption{Waves propagating from the antenna placed at an angle $\theta = \pi/8$ with the vertical axis, toward the mode conversion point $(0,0)$. The parameter $d$ is $d=-2-\imath$.   Left: Computed for $y_K=7$. Right: Computed for $y_K=12$. }
\label{fig:thetamain}
\end{figure}

The transmission process is expected to depend strongly on the incident angle $\theta$. Figure \ref{fig:Tthetas} displays this transmission coefficient as a function of the incident angle $\theta$, for both $d=-2-\imath$ and $d = -tan(\pi/8)-\imath$. Typical values of various geometry and mesh parameters used to produce Figure \ref{fig:Tthetas} are given in Table \ref{tab:meshparam}. As expected, the transmission is very low when $\theta$ is close to zero and $\pi/4$, which correspond to launching the wave parallel one of the cut-offs, while a peak is observed between those two regimes, which corresponds to the wave propagation direction closer to the steepest gradient line of the variable coefficients. Figure \ref{fig:maxmin} displays the solution with the maximum transmission as well as a solution completely reflected by the cut-off before reaching the mode conversion region. 
\begin{figure}
\begin{center}
\includegraphics[clip,trim=2cm 1cm 1.5cm 1cm,height=4cm]{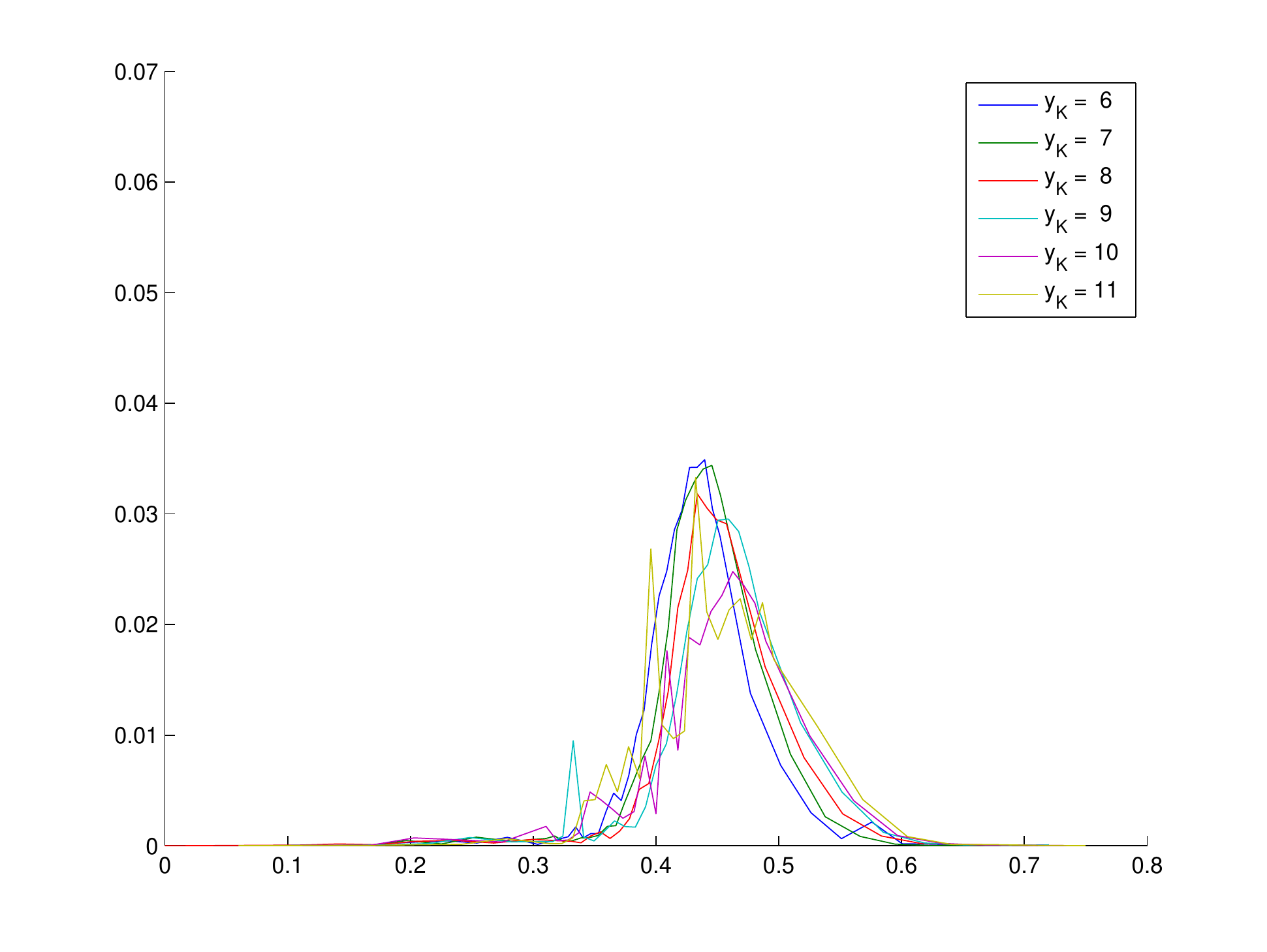}
\includegraphics[clip,trim=2cm 1cm 1.5cm 1cm,height=4cm]{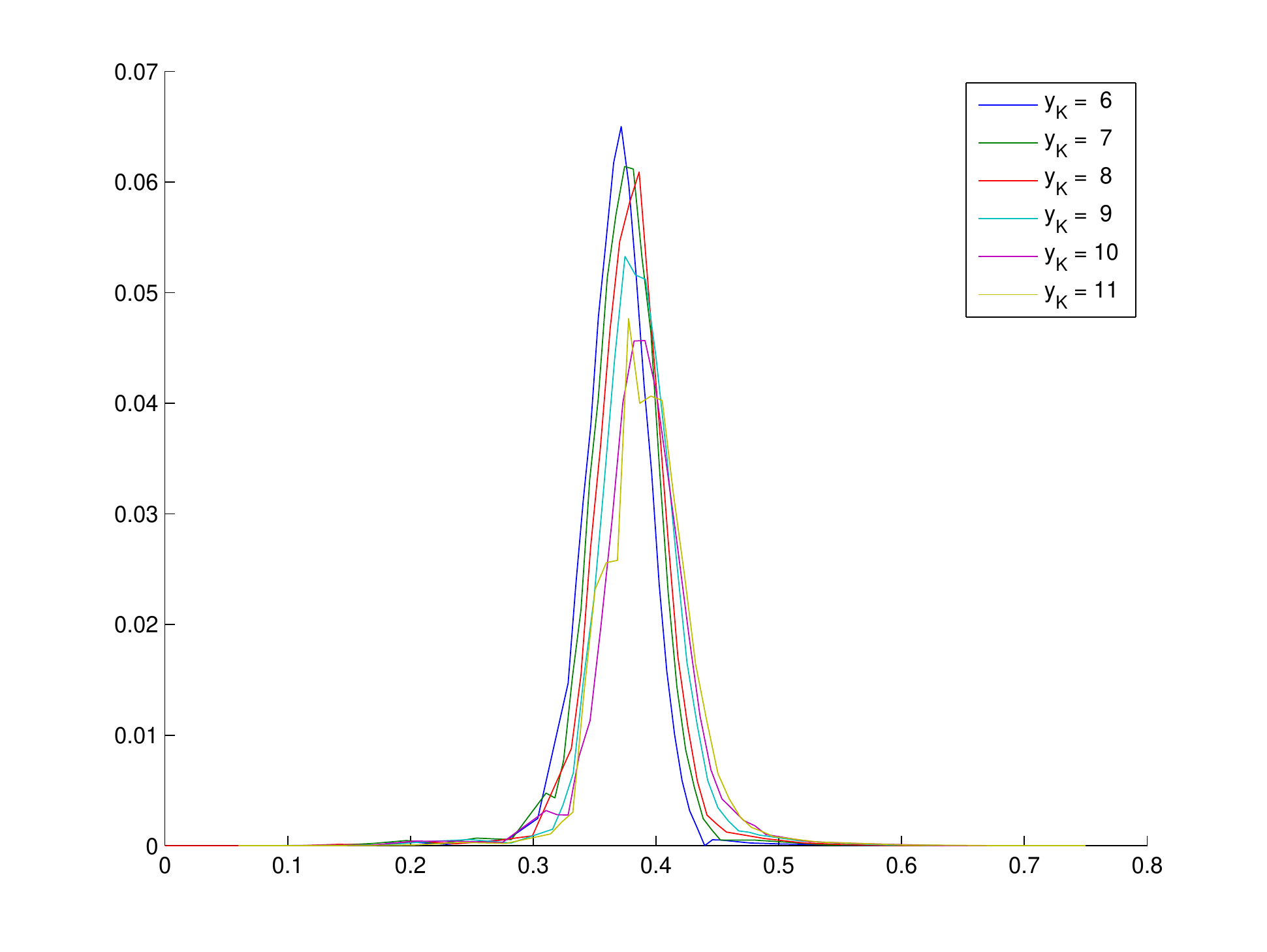}
\end{center}
\caption{Transmission coefficient as a function of the incident angle $\theta$, for different sizes of transition zones. Left: Computed for $d = -2-\imath$. Right: Computed for $d = -\tan(3\pi/8)-\imath$. }
\label{fig:Tthetas}
\end{figure}  
\begin{table}[h!]
  \begin{center}
    \caption{Geometry and mesh parameters for different sizes of the transition zone.}
    \label{tab:meshparam}
    \begin{tabular}{c|cccccc}
     $y_K$ & 6 & 7 & 8 & 9 & 10 & 11\\
      \hline
      $l_0$ & 2.13 & 1.82 & 1.59 & 1.42 & 1.28 & 1.16\\
      mesh size  & 0.73 &0.64 & 0.57 & 0.46 & 0.43 & 0.39 \\
      number of triangles & 25157 & 36548 & 42412 & 60944 & 66511 &  78470 \\
    \end{tabular}
  \end{center}
\end{table}
\begin{figure}
\begin{center}
\includegraphics[clip,trim=1cm .5cm 1cm .5cm,height=4cm]{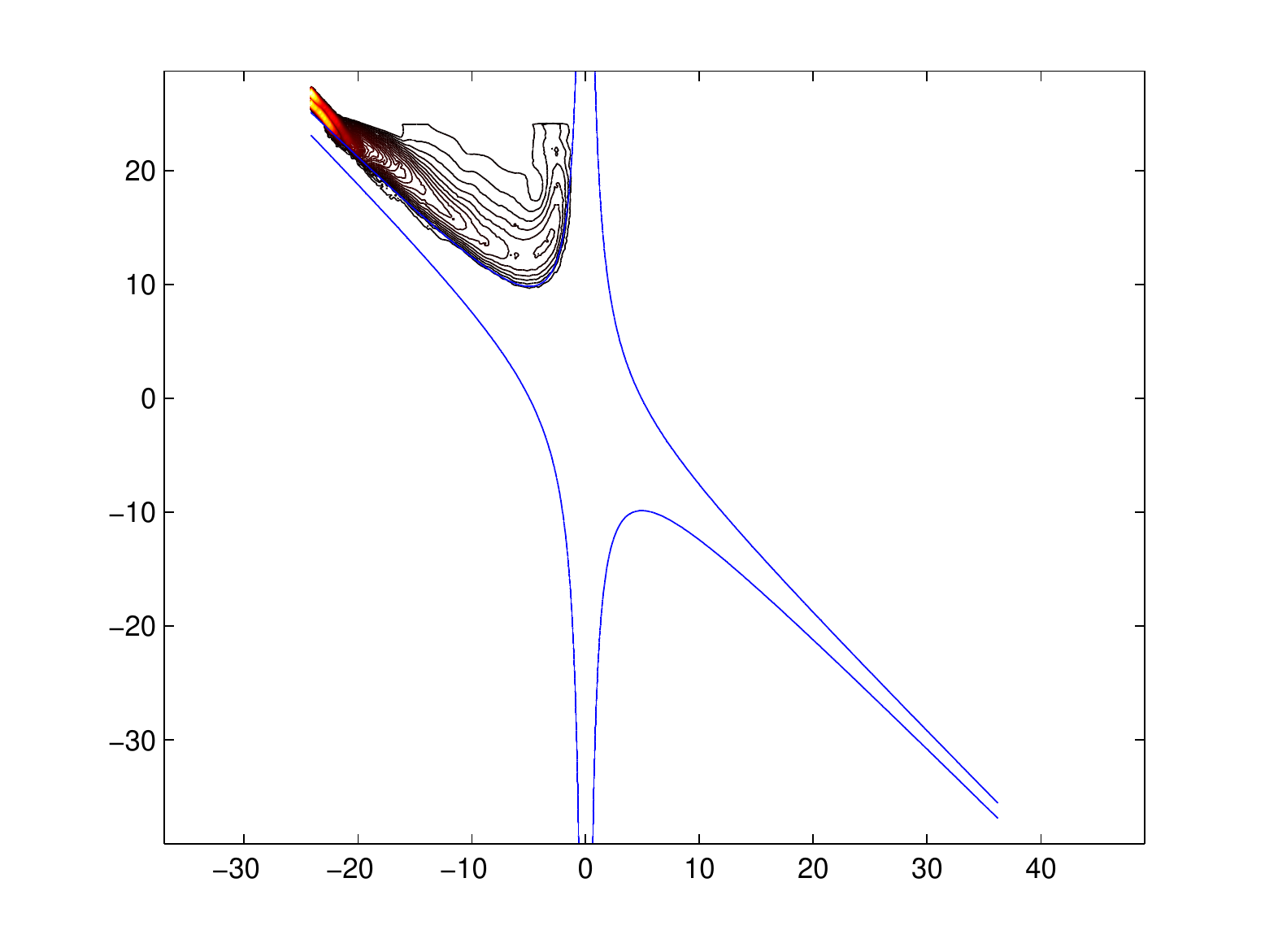}
\includegraphics[clip,trim=1cm .5cm 1cm .5cm,height=4cm]{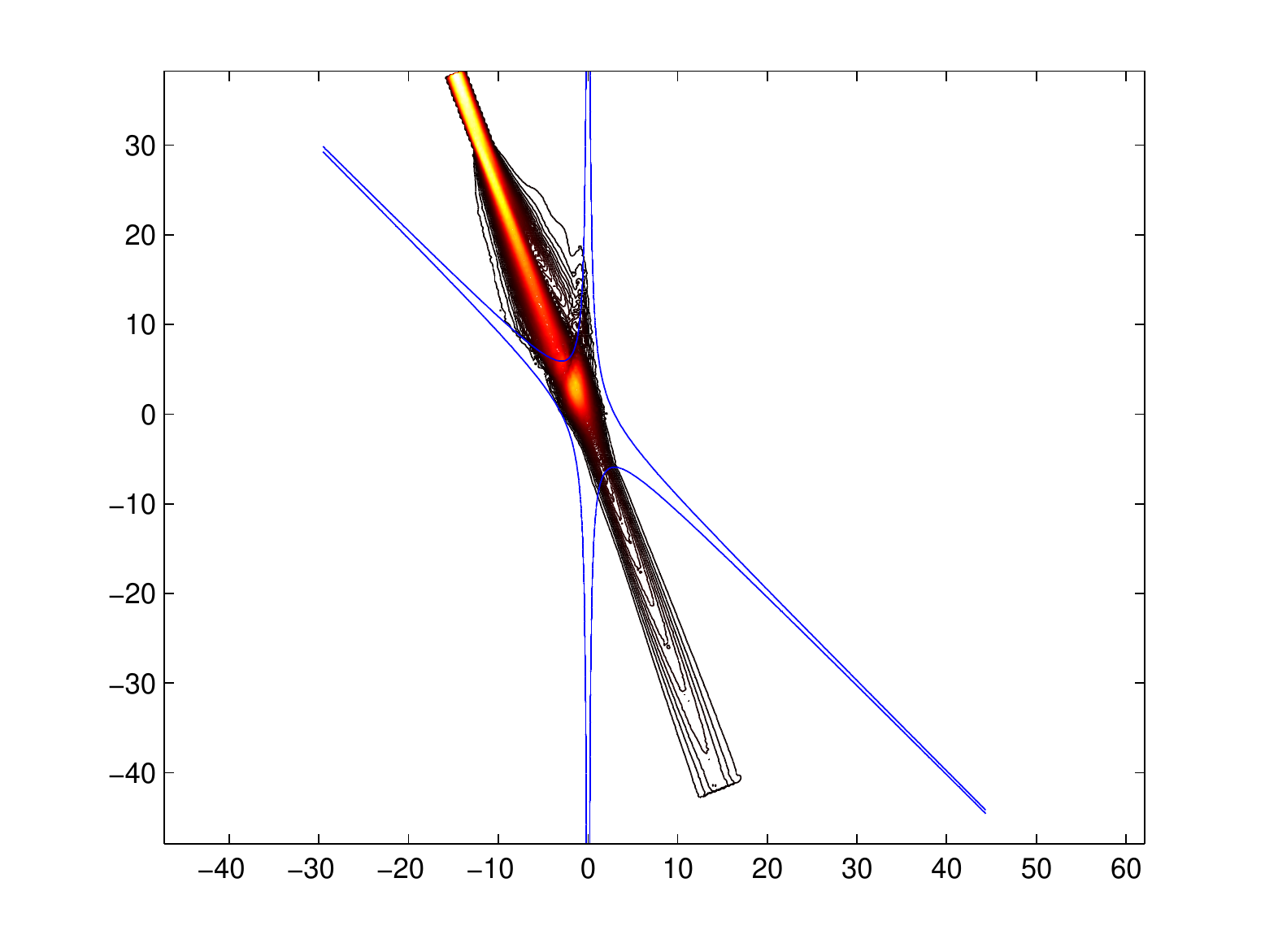}
\end{center}
\caption{Comparison of two solutions, highlighting the influence of the incident angle on the transmission coefficient. 
Left: Computed for $y_K = 10$, $d = -2-\imath$ and $\theta = 0.74$, corresponding to a transmission coefficient $T\approx  10^{-5}$. 
Right: Computed for $y_K = 6$, $d = -\tan(3\pi/8)-\imath$ and $\theta = 0.371$, corresponding to the transmission coefficient $T=0.065$.
}
\label{fig:maxmin}
\end{figure}  

Figure \ref{fig:thetamain} evidences the importance of the size of the transition zone on the mode conversion efficiency. Indeed the maximum transmission coefficient is a decreasing function of $y_K$, which is a measure of the size of the transition zone. Moreover the size of the transition zone also influences the angle at which this maximum coefficient occurs. See Tables \ref{tab:2} and \ref{tab:3}.

\begin{table}[ht]
\caption{
Comparison of maximum transmission coefficient and corresponding incident angle, for different sizes of variable coefficient zone. Computed for $d = -2-\imath$.}
{\begin{tabular}{lcccccc}\hline
$y_K$ & 6 & 7 & 8 & 9 & 10 & 11  \\ \hline\hline
 $T_{\max}(y_K)$ & 3.49e-02 & 3.44e-02 & 3.18e-02 & 2.95e-02 & 2.48e-02 & 3.33e-02  \\ 
$\overline \theta(y_K)$ & 0.434 & 0.444 & 0.451 & 0.460 & 0.465 & 0.458 \\ \hline
\end{tabular}}
\label{tab:2}
\end{table}

\begin{table}[ht]
\caption{
Comparison of maximum transmission coefficient and corresponding incident angle, for different sizes of variable coefficient zone. Computed for $d = -\tan(3\pi/8)-\imath$.}
{\begin{tabular}{lcccccc} \hline
$y_K$ & 6 & 7 & 8 & 9 & 10 & 11  \\   \hline\hline
$T_{\max}(y_K)$ & 6.50e-02 & 6.14e-02 & 6.09e-02 & 5.33e-02 & 4.57e-02 & 4.76e-02  \\ 
$\overline \theta(y_K)$ & 0.370 & 0.375 & 0.380 & 0.384 & 0.390 & 0.391 
\end{tabular}}
\label{tab:3}
\end{table}

\section*{Acknowledgment}
Many thanks to Harold Weitzner for introducing me to mode conversion and guiding me along the way, 
and to Jonathan Goodman for our many discussions. I also thank Teemu Luostari for providing his 2D PW-UWVF code for elasticity equations.

This work was supported in part by the U.S. Department of Energy, Office of Sci-
ence, under Grant No. DE-FG02-86ER53223.


\end{document}